\documentclass[onecolumn]{IEEEtran}
\usepackage{amsthm, amsmath, amssymb, amsfonts, url, booktabs, tikz, setspace, fancyhdr, bm}
\usetikzlibrary{decorations.pathreplacing}
\usepackage{hyperref}
\usepackage{amsthm}
\usepackage{cancel}
\usepackage{hyperref, enumerate}
\usepackage[shortlabels]{enumitem}
\usepackage[babel]{microtype}
\usepackage[english]{babel}
\usepackage[capitalise]{cleveref}
\usepackage{comment}
\usepackage{bbm}
\usepackage{csquotes}
\usepackage{mathabx}
\usepackage{algorithm,algpseudocode,algorithmicx}
\usepackage{tikz}
\usetikzlibrary{shapes.geometric, positioning}
\usepackage{subcaption}
\usepackage{graphicx}
\usepackage{float}
\usepackage{xcolor}
\usepackage{tikz-cd}
\usepackage{adjustbox}
\usepackage{mathrsfs}
\usepackage{pgfplots}
\pgfplotsset{compat=newest}
\usetikzlibrary{decorations.pathmorphing}
\usetikzlibrary{positioning,arrows.meta, shapes.geometric}

\counterwithin{figure}{section}

\newtheorem{theorem}{Theorem}[section]
\newtheorem{corollary}[theorem]{Corollary}

\newtheorem{lemma}[theorem]{Lemma}
\newtheorem{example}[theorem]{Example}

\newtheorem{remark}[theorem]{Remark}

\newtheorem{definition}[theorem]{Definition}

\algrenewcommand\algorithmicrequire{\textbf{Input:}}
\algrenewcommand\algorithmicensure{\textbf{Output:}}
\algnewcommand{\LineComment}[1]{\State $\triangleright$ #1}

\title{On the Palindromic/Reverse-Complement Duplication Correcting Codes}

\author{Yubo~Sun and Gennian~Ge
\thanks{This research was supported by the National Key Research and Development Program of China under Grant 2025YFC3409900, the National Natural Science Foundation of China under Grant 12231014, and Beijing Scholars Program.}
\thanks{Y. Sun ({\tt 2200502135@cnu.edu.cn}) and G. Ge ({\tt gnge@zju.edu.cn}) are with the School of Mathematical Sciences, Capital Normal University, Beijing 100048, China.}
}

\begin{document}

\date{}
\maketitle

\begin{abstract}
    Motivated by applications in in-vivo DNA storage, we study codes for correcting duplications. A reverse-complement duplication of length $k$ is the insertion of the reversed and complemented copy of a substring of length $k$ adjacent to its original position, while a palindromic duplication only inserts the reversed copy without complementation. We first construct an explicit code with a single redundant symbol capable of correcting an arbitrary number of reverse-complement duplications (respectively, palindromic duplications), provided that all duplications have length $k \ge 3\lceil \log_q n \rceil$ and are disjoint. Next, we derive a Gilbert–Varshamov bound for codes that can correct a reverse-complement duplication (respectively, palindromic duplication) of arbitrary length, showing that the optimal redundancy is upper bounded by $2\log_q n + \log_q\log_q n + O(1)$. Finally, for $q \ge 4$, we present two explicit constructions of codes that can correct $t$ length-one reverse-complement duplications. The first construction achieves a redundancy of $2t\log_q n + O(\log_q\log_q n)$ with encoding complexity $O(n)$ and decoding complexity $O\big(n(\log_2 n)^4\big)$. The second construction achieves an improved redundancy of $(2t-1)\log_q n + O(\log_q\log_q n)$, but with encoding and decoding complexities of $O\big(n \cdot \mathrm{poly}(\log_2 n)\big)$.
\end{abstract}

\section{Introduction}

In the current era of big data, the global volume of digital data is growing exponentially. However, the storage capacity of existing technologies is unlikely to keep pace with the rising demand in the foreseeable future. This discrepancy necessitates breakthroughs in storage technology to deliver innovative, high-density, efficient, and durable solutions. Among these, DNA storage systems—with their ultra-high density, exceptional durability, and low maintenance costs—emerge as a compelling alternative to traditional storage media.
The feasibility of storing information in DNA storage was demonstrated in \cite{Church-12-Science,Goldman-13-nature} (for in-vitro DNA storage) and \cite{Shipman-17-Nature} (for in-vivo DNA storage). Compared to in-vitro DNA storage, in-vivo DNA storage offers additional functionalities, such as the ability to watermark genetically modified organisms (GMOs) to verify authenticity and track unauthorized use \cite{Arita-04,Heider-07-BMS,Liss-12}, as well as to label organisms in biological studies \cite{Wong-03}.
However, in-vivo DNA storage introduces a unique challenge in the form of duplication errors, which can occur during faulty DNA replication processes. This motivates us to investigate codes capable of correcting such errors.

Generally, a duplication error of length $k$ involves taking a substring of length $k$ of the DNA string and inserting its copy (possibly altered) adjacent to its original location. There are three primary types of duplication errors:
\begin{itemize}
\item Tandem duplication, where the copy remains unaltered;
\item Palindromic duplication (also known as reverse duplication), where the copy is reversed;
\item Reverse-complement duplication, where the copy is first reversed and then complemented.
\end{itemize}
While tandem duplications have been extensively studied in the literature \cite{Elishco-24-IT,Elishco-19-IT,Farnoud-16-IT,Goshkoder-24-TMBMC,Jain-17-IT', Jain-17-IT,Kovačević-19-TCOM,Kovačević-22-PIT,Kovačević-18-CL,Tang-22-IT,Tang-23-IT,Tang-20-IT,Yu-24-DCC}, reverse-complement duplications and palindromic duplications remain relatively underexplored due to their intricate intrinsic structures.

In this paper, we focus on reverse-complement duplications and palindromic duplications.
Recognizing that duplications can be viewed as a special type of burst-insertions, one can simply use burst-insertions correcting codes to correct duplications. 
However, this approach is not the most efficient, as duplications inherently provide additional information compared to generic burst-insertions.
A natural question arises: Can we design codes that correct duplications with less redundancy compared to general burst-insertions correcting codes \cite{Guruswami-21-IT,Levenshtein-66,Levenshtein-70,Li-23-ISIT, Sima-20-ISIT,Sima-20-ISIT', Song-23-IT, Song-22-IT,Sun-24-IT, Sun-25-IT, Ye-25-IT} by leveraging the intrinsic properties of duplications?
In \cite{Ben-Tolila-22-IT}, Ben-Tolila and Schwartz demonstrated that when the duplication length $k$ is odd and the alphabet size $q \geq 2$, any binary code capable of correcting a burst-insertion of length $k$ can be utilized through the complement preserving mapping to construct $q$-ary codes that can correct a reverse-complement duplication of length $k$. 
Notably, when $q \geq 4$, they derived constructions with reduced redundancy compared to general single burst-insertion correcting codes.
Furthermore, Yohananov and Schwartz \cite{Yohananov-25-DCC} constructed optimal codes with rates $\log_q(q-2)$ and $\log_q(q-1)$, capable of correcting an arbitrary number of reverse-complement duplications and palindromic duplications, respectively, of length $k=1$.
In \cite{Lenz-19-DCC}, Lenz, Wachter-Zeh, and Yaakobi established that if a code $\mathcal{C}$ does not contain palindromic repeats of length at least two, i.e., no two adjacent substrings of length $k\geq 2$ such that the second is a reverse of the first, then $\mathcal{C}$ can correct a palindromic duplication of any length $k \geq 2$. 
Additionally, they proved the existence of such codes with a rate of at least $\log_q(q-1)$.
Recent work by Liu, Tang, Fan, and Sagar \cite{Liu-25-AMC} extended this conclusion to reverse-complement duplications.
Given that the rate of codes that can correct an arbitrary number of burst-insertions of length $k=1$ or a burst-insertion of arbitrary length $k \geq 2$ is zero, the works referenced in \cite{Lenz-19-DCC,Liu-25-AMC,Yohananov-25-DCC} also derived codes that perform better in correcting duplications than general burst-insertions correcting codes.

The rest of this paper is organized as follows.
In Section \ref{sec:pre}, we introduce the relevant notations used throughout the paper.
In Section \ref{sec:long}, we consider codes for correcting palindromic/reverse-complement duplications of long length. We show that if a code $\mathcal{C}$ avoids reverse-complement repeats (respectively, palindromic repeats) of length $m$, then $\mathcal{C}$ can correct arbitrary $t\geq 1$ reverse-complement duplications (respectively, palindromic duplications), provided that all duplications have length $k\geq 3m-3$ and are disjoint.
For $m=2$ and $t=1$, this code can be generalized to correct a reverse-complement duplication (respectively, palindromic duplication) of length at least two, thereby recovering the results established in \cite{Lenz-19-DCC} and \cite{Liu-25-AMC}.
For $m = \lceil \log_q n \rceil + 1$, we propose efficient encoding and decoding algorithms with one redundant symbol for this code.
In Section \ref{sec:arbitrary}, we consider codes for correcting a reverse-complement duplication (respectively, palindromic duplication) of arbitrary length. In conjunction with the codes developed in Section \ref{sec:long}, we derive a Gilbert–Varshamov bound for such codes, showing that there exists a construction with $2\log_q n + \log_q\log_q n + O(1)$ redundant symbols, i.e., with asymptotic rate $1$. This improves the rate bound $\log_q(q-1)$ established in \cite{Lenz-19-DCC} and \cite{Liu-25-AMC}.
In Section \ref{sec:short}, we consider codes for correcting multiple length-one reverse-complement duplications. By adopting the techniques from \cite{Yohananov-25-DCC}, we transform length-one reverse-complement duplications into substitutions. Moreover, by imposing run-length-limited constraints and combining a substitution-correcting code from \cite{Liu-24-ISIT} with an indel-correcting code from \cite{Li-23-ISIT}, we design two families of $q$-ary codes (with $q \ge 4$) that can correct $t$ length-one reverse-complement duplications.
The first construction achieves a redundancy of $2t\log_q n + O(\log_q\log_q n)$ with encoding complexity $O(n)$ and decoding complexity $O\big(n(\log_2 n)^4\big)$. The second construction achieves an improved redundancy of $(2t-1)\log_q n + O(\log_q\log_q n)$, but with encoding and decoding complexities of $O\big(n \cdot \mathrm{poly}(\log_2 n)\big)$. These constructions outperform the approach of directly using indel-correcting codes \cite{Li-23-ISIT, Sima-20-ISIT', Song-23-IT, Song-22-IT, Ye-25-IT}, which require $5\log_q n + O(\log_q\log_q n)$ and $(4t-1)\log_q n + O(\log_q\log_q n)$ redundant symbols when $t=2$ and $t \ge 3$, respectively. Finally, in Section \ref{sec:concl}, we conclude the paper.

\section{Preliminaries}\label{sec:pre}

\subsection{Notations}

Given two integers $i,j$, define the integer interval $[i, j]$ as $\{i, i+1, \ldots, j\}$ if $i \leq j$, and as the empty set $\varnothing$ otherwise. 
For any integer $q \ge 2$, define $\Sigma_q = [0, q-1]$ as the $q$-ary alphabet. 
Let $\Sigma_q^n$ denote the set of all strings of length $n$ over $\Sigma_q$. 
For any string $\boldsymbol{x} \in \Sigma_q^n$, we write either $\boldsymbol{x} = x_1 x_2 \cdots x_n$ or $\boldsymbol{x} = (x_1, x_2, \ldots, x_n)$, where $x_s$ denotes the $s$-th symbol of $\boldsymbol{x}$ for $s\in [1, n]$. 
The \emph{length} of $\boldsymbol{x}$ is denoted as $|\boldsymbol{x}|$.
For any integer interval $[i,j] \subseteq [1,n]$, we say that $\boldsymbol{x}_{[i,j]}=x_i\cdots x_j$ is a \emph{substring} of $\boldsymbol{x}$, where $x_i\cdots x_j$ denotes the empty string $\boldsymbol{\epsilon}$ when $i>j$.
Let $\boldsymbol{y} \in \Sigma_q^m$. The \emph{concatenation} of $\boldsymbol{x}$ and $\boldsymbol{y}$, denoted as either $\boldsymbol{x} \boldsymbol{y}$ or $(\boldsymbol{x}, \boldsymbol{y})$, is defined by $x_1 x_2 \cdots x_n y_1 y_2 \cdots y_m$. 
For a set $V$, its \emph{cardinality} is denoted as $\#V$.
For a non-negative integer $a$, let $Rep_{q,m}(a)$ be the $q$-ary representation of $a$ of length $m$.
Then its inverse mapping $Rep_{q,m}^{-1}(\cdot)$ defined by $Rep_{q,m}^{-1}\big(Rep_{q,m}(a)\big)=a$ is well-defined.

\subsection{Complement and Reverse Operations}

A \emph{complement operation}\footnote{The complement operation can be viewed as a permutation $\pi:\Sigma_q \to \Sigma_q$ with $\pi(a)\neq a$ and $\pi\big(\pi(a)\big)=a$ for each $a\in\Sigma_q$.} is characterized as a bijective map that assigns each symbol $a$ in the alphabet $\Sigma_q$ to a unique symbol $\overline{a}$ in $\Sigma_q$, adhering to the following properties:
\[
\overline{a}\neq a \quad \text{and} \quad \overline{\overline{a}} = a \quad \text{for each } a \in \Sigma_q.
\]
An implicit assumption in this setting is that $q$ is even. Therefore, when we consider the complement operation, we always assume that $q$ is even.
This complement operation extends naturally to entire strings, i.e., for $\boldsymbol{x}=x_1x_2\cdots x_n$, we define
\[
\overline{\bm{x}} = \overline{x_1}\, \overline{x_2} \cdots \overline{x_n}.
\]
The \emph{reverse} of $\bm{x}$, denoted as $\bm{x}^R$, is constructed by reversing the order of $\bm{x}$, i.e., we define
\[
\bm{x}^R = x_n x_{n-1} \cdots x_1.
\]
Additionally, the \emph{reverse-complement} of $\bm{x}$, denoted as $\bm{x}^{RC}$, is obtained by reversing $\bm{x}$ followed by taking the complement of each symbol, i.e., we define
\[
\bm{x}^{RC} = \overline{x_n}\, \overline{x_{n-1}} \cdots \overline{x_1}.
\]

\begin{remark}
    The complement and reverse operations commute. Specifically, it holds that
    \[
    \bm{x}^{RC} = \overline{\bm{x}^R} = \overline{\bm{x}}^R.
    \]
\end{remark}

\subsection{Duplication Correcting Codes}
 
Now we formalize the definition of reverse-complement duplications and palindromic duplications.
Let $\bm{x} = \bm{u} \bm{v} \bm{w}$ be a string of length $n$, where $\bm{u}$, $\bm{v}$, and $\bm{w}$ are substrings with lengths $|\bm{u}| = i - 1$, $|\bm{v}| = k$, and $|\bm{w}| = n - i - k + 1$, respectively.
An error at position $i \in [1, n - k + 1]$ in $\bm{x}$ is termed a \emph{$k$-reverse-complement duplication} if the resultant string is
\[
RC_{k,i}(\bm{x}) = \bm{u} \bm{v} \bm{v}^{RC} \bm{w}.
\]
Similarly, an error at position $i \in [1, n - k + 1]$ in $\bm{x}$ is termed a \emph{$k$-palindromic duplication} if the resultant string is
\[
R_{k,i}(\bm{x}) = \bm{u} \bm{v} \bm{v}^{R} \bm{w}.
\]
The \emph{$k$-reverse-complement-duplication ball centered at $\bm{x}$ of radius $t$} is defined as 
\[
RC_{k}^t(\bm{x}) = \left\{RC_{k,i_t}\big(RC_{k,i_{t-1}}(\ldots(RC_{k,i_1}(\bm{x})\ldots) \big):i_j\in [1,n+(j-1)k-k+1] \text{ for }j\in [1,t]\right\}.
\]
A code $\mathcal{C} \subseteq \Sigma_q^n$ is termed a \emph{$t$ $k$-reverse-complement-duplications correcting code} if it can correct $t$ $k$-reverse-complement duplications, i.e.,
\[
RC_{k}^t(\bm{x})\cap RC_{k}^t(\bm{y})=\varnothing \quad \text{for each } \bm{x}\neq \bm{y} \in \mathcal{C}.
\]
The definition of \emph{$t$ $k$-palindromic-duplications correcting code} can be defined analogously.
To evaluate a code $\mathcal{C} \subseteq \Sigma_q^n$, we consider either its \emph{redundancy}, defined as $n - \log_q |\mathcal{C}|$, or its \emph{rate}, defined as $\frac{\log_q |\mathcal{C}|}{n}$. 

\section{Correcting Long Palindromic/Reverse-Complement Duplications}\label{sec:long}

For tandem duplications, Jain, Farnoud, Schwartz, and Bruck \cite{Jain-17-IT} and Goshkoder, Polyanskii, Vorobyev \cite{Goshkoder-24-TMBMC} established that a code $\mathcal{C}$ containing sequences with no tandem repeats of length $k$ inherently possesses the ability to correct a tandem duplication of that length.
However, this result does not extend universally to reverse-complement duplications and palindromic duplications, despite its validity for duplications of length $k=2$, as demonstrated in \cite{Lenz-19-DCC} and \cite{Liu-25-AMC}.
To illustrate this limitation, we consider the code $\mathcal{C} = \{000111000, 000111110\}$, which is specifically designed to exclude any reverse-complement repeats of length four.
Despite this design, $\mathcal{C}$ fails to correct a reverse-complement duplication of length four. Specifically:
\begin{itemize}
  \item The string $00011\underline{1000}$ can evolve into $0001110001110$, where the underlined substring denotes the duplicated portion.
  \item Moreover, the string $00\underline{0111}110$ can also evolve into $0001110001110$.
\end{itemize}
This observation underscores the necessity for a more nuanced approach when addressing reverse-complement duplications and palindromic duplications.

In what follows, we focus solely on reverse-complement duplications, since all results established herein can be seamlessly extended to handle palindromic duplications.
We first show that if a code $\mathcal{C}$ ensures that none of its codewords contains two adjacent substrings of length $m$ such that the second is a reverse-complement of the first, then $\mathcal{C}$ can correct a reverse-complement duplication of length at least $3m-3$. 
To facilitate this, we introduce the following concept.

\begin{definition}\label{def:SSA}
  Let $m\geq 2$ be an integer.
  A string $\bm{x} \in \Sigma_q^n$ is termed an \emph{$m$-reverse-complement-duplication root} (or \emph{$m$-RCD root} for short) if, for any two adjacent substrings of length $m$, the second is not a reverse-complement of the first. In other words, for each $i\in [1,n - 2m + 1]$, we require
    \[
    \bm{x}_{[i+m, i+2m-1]}\neq \bm{x}_{[i, i+m-1]}^{RC},
    \]
    or equivalently
    \[
    \bm{x}_{[i, i+m-1]} \neq \bm{x}_{[i+m, i+2m-1]}^{RC}.
    \]
\end{definition}

The following conclusion is important for our code design:

\begin{lemma}\label{lem:SSA}
    Let $\bm{x} \in \Sigma_q^n$ be an $m$-RCD root, where $m\geq 2$ is an integer. 
    If $\bm{y} = RC_{k,i}(\bm{x})$ for some integer $k \geq 3m-3$ and position $i \in [1, n - k + 1]$, then $\bm{y}_{[i+k, i+2k-1]}=\bm{y}_{[i, i+k-1]}^{RC}$ and $\bm{y}_{[j+k, j+2k-1]}\neq \bm{y}_{[j, j+k-1]}^{RC}$ for $j < i$.
\end{lemma}

\begin{proof}
    Since $\bm{y}=RC_{k,i}(\bm{x})$, it follows by definition that $\bm{y}_{[i+k,i+2k-1]}=\bm{y}_{[i,i+k-1]}^{RC}$ and $\bm{y}_{[1,i+k-1]}=\bm{x}_{[1,i+k-1]}$.
    Given that $\bm{x}$ is an $m$-RCD root, we see that the prefix $\bm{y}_{[1, i+k-1]}$ also maintains the $m$-RCD root property.
    Suppose, for contradiction, that there exists some $j < i$ such that $\bm{y}_{[j + k, j + 2k - 1]}=\bm{y}_{[j, j+k-1]}^{RC}$.
    We analyze two cases:
    \begin{itemize}
        \item If $j\leq i-m$, then $j+k+m-1\leq i+k-1$. 
        Since $\bm{y}_{[j+k,j+2k-1]}=\bm{y}_{[j,j+k-1]}^{RC}$, we get $\bm{y}_{[j+k,j+k+m-1]}=\bm{y}_{[j+k-m,j+k-1]}^{RC}$, which contradicts the fact that $\bm{y}_{[1,i+k-1]}$ is an $m$-RCD root.
        
        \item If $i-m<j<i$, let $\ell\triangleq |\bm{x}_{[j,i-1]}|=i-j$ and $s=\lfloor \frac{k+\ell}{2\ell} \rfloor$, then $\ell\in [1,m-1]$ and $2s\ell\in [k-\ell+1,k+\ell]$ is an even integer.
        Note that $k-\ell+1\geq (3m-3)-(m-1)+1=2m-1$, we get $2s\ell\geq 2m$, i.e., $s\ell\geq m$.
        We define $\bm{u}^{(t)}= \bm{y}_{[i+k-t\ell, i+k-(t-1)\ell-1]}$ for $t\in [1,2s]$ and $\bm{v}^{(t)}= \bm{y}_{[i+k+(t-1)\ell, i+k+t\ell-1]}$ for $t\in [1,2s-1]$, and illustrate the definitions in Figure \ref{fig:uv}.
        Since $\bm{y}_{[i+k,i+2k-1]}=\bm{y}_{[i,i+k-1]}^{RC}$, we have $\bm{v}^{(t)}=\big(\bm{u}^{(t)} \big)^{RC}$ for $t\in [1,2s-1]$.
        Moreover, since $\bm{y}_{[j+k,j+2k-1]}=\bm{y}_{[j,j+k-1]}^{RC}$, we have $\bm{u}^{(1)}=\big(\bm{u}^{(2)} \big)^{RC}$ and $\bm{v}^{(t)}=\big(\bm{u}^{(t+2)} \big)^{RC}$ for $t\in [1,2s-2]$.
        It follows that $\bm{u}^{(1)}=\bm{u}^{(3)}=\cdots=\bm{u}^{(2s-1)}= \big(\bm{u}^{(2)} \big)^{RC}$ and $\bm{u}^{(2)}=\bm{u}^{(4)}=\cdots=\bm{u}^{(2s)}$.
        Then, we get 
        \begin{align*}
            \bm{y}_{[i+k-s\ell,i+k-1]}
            &=\bm{u}^{(s)} \bm{u}^{(s-1)} \cdots \bm{u}^{(1)}\\
            &= \big(\bm{u}^{(s+1)}\big)^{RC} \big(\bm{u}^{(s+2)}\big)^{RC} \cdots \big(\bm{u}^{(2s)} \big)^{RC}\\
            &= \big(\bm{u}^{(2s)} \bm{u}^{(2s-1)} \cdots \bm{u}^{(s+1)}\big)^{RC}\\
            &= \bm{y}_{[i+k-2s\ell,i+k-s\ell-1]}^{RC}.
        \end{align*}
        Note that the length of $\bm{y}_{[i+k-s\ell,i+k-1]}$ is $s\ell\geq m$,
        we obtain a  contradiction with the fact that $\bm{y}_{[1,i+k-1]}$ is an $m$-RCD root.
    \end{itemize}
    In both cases, assuming the existence of such a $j < i$ leads to a contradiction, so no such $j$ can exist. 
    This completes the proof.
\end{proof}

\begin{figure*}
    \centering
    \begin{tikzpicture}
        \draw[thick, black] (0,0) -- (15,0) node[right] {$\bm{y}$};

        \draw[thick, black, dashed] (0.7,0.2) -- (0.7,-1.5) node[below=0.5pt] {$j$};
        
        \draw[thick, black, dashed] (1.5,0.2) -- (1.5,-1) node[below=0.5pt] {$i-1$};
        
        \draw[decorate,decoration={brace,mirror,amplitude=6pt,raise=4pt}] (0.7,-0.2) -- (1.5,-0.2);
        \node[below=0.5pt] at (1.1,-0.5) {$\ell$};

        \draw[thick, black, dashed] (1.7,0.2) -- (1.7,-0.5) node[below=0.5pt] {$i$};

        \draw[thick, black, dashed] (7,0.2) -- (7,-1.5) node[below=0.5pt] {$j+k-1$};

        \draw[thick, black, dashed] (8,0.2) -- (8,-1) node[below=0.5pt] {$i+k-1$};
        
        \draw[thick, black, dashed] (8.2,0.2) -- (8.2,-0.5);
        \node[below=0.5pt] at (8.55,-0.5) {$i+k$};

        \draw[thick, black, dashed] (13.3,0.2) -- (13.3,-1.5) node[below=0.5pt] {$j+2k-1$};

        \draw[thick, black, dashed] (14.3,0.2) -- (14.3,-1) node[below=0.5pt] {$i+2k-1$};

        \node[draw, rectangle, minimum width=0.8cm, minimum height=0.5cm] at (1.6, 0) {};
        \node[above=0.5pt] at (1.7,0.3) {$\bm{u}^{(2s)}$};
        
        \node[draw, rectangle, minimum width=0.8cm, minimum height=0.5cm] at (2.6, 0) {};
        \node[above=0.5pt] at (2.7,0.3) {$\bm{u}^{(2s-1)}$};
        
        \node[above=0.5pt] at (3.6,0) {$\cdots$};
        
        \node[draw, rectangle, minimum width=0.8cm, minimum height=0.5cm] at (4.6, 0) {};
        \node[above=0.5pt] at (4.7,0.3) {$\bm{u}^{(4)}$};
        \draw[decorate,decoration={brace,mirror,amplitude=6pt,raise=4pt}] (4.2,-0.2) -- (5,-0.2);
        \node[below=0.5pt] at (4.6,-0.5) {$\ell$};
        
        \node[draw, rectangle, minimum width=0.8cm, minimum height=0.5cm] at (5.6, 0) {};
        \node[above=0.5pt] at (5.7,0.3) {$\bm{u}^{(3)}$};
        
        \node[draw, rectangle, minimum width=0.8cm, minimum height=0.5cm] at (6.6, 0) {};
        \node[above=0.5pt] at (6.7,0.3) {$\bm{u}^{(2)}$};
        
        \node[draw, rectangle, minimum width=0.8cm, minimum height=0.5cm] at (7.6, 0) {};
        \node[above=0.5pt] at (7.7,0.3) {$\bm{u}^{(1)}$};
        
        \node[draw, rectangle, minimum width=0.8cm, minimum height=0.5cm] at (8.6, 0) {};
        \node[above=0.5pt] at (8.7,0.3) {$\bm{v}^{(1)}$};
        
        \node[draw, rectangle, minimum width=0.8cm, minimum height=0.5cm] at (9.6,0) {};
        \node[above=0.5pt] at (9.7,0.3) {$\bm{v}^{(2)}$};
        
        \node[draw, rectangle, minimum width=0.8cm, minimum height=0.5cm] at (10.6,0) {};
        \node[above=0.5pt] at (10.7,0.3) {$\bm{v}^{(3)}$};
        
        \node[draw, rectangle, minimum width=0.8cm, minimum height=0.5cm] at (11.6,0) {};
        \node[above=0.5pt] at (11.7,0.3) {$\bm{v}^{(4)}$};
        
        \node[above=0.5pt] at (12.6,0) {$\cdots$};
        
        \node[draw, rectangle, minimum width=0.8cm, minimum height=0.5cm] at (13.6,0) {};
        \node[above=0.5pt] at (13.7,0.3) {$\bm{v}^{(2s-1)}$};

    \end{tikzpicture}
    \caption{Illustrations of the definitions of $\bm{u}^{(t)}$ for $t\in [1,2s]$ and $\bm{v}^{(t)}$ for $t\in [1,2s-1]$.}
    \label{fig:uv}
\end{figure*}
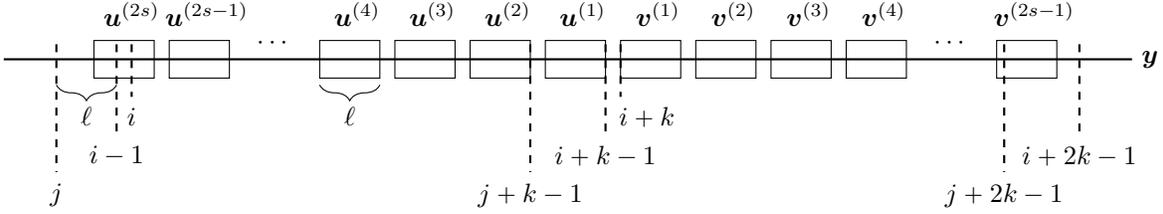

Suppose an $m$-RCD root $\bm{x}\in \Sigma_q^n$ suffers a reverse-complement duplication of length $k\geq 3m-3$ and results in $\bm{y}$, let $i$ be the smallest index such that $\bm{y}_{[i+k,i+2k-1]}=\bm{y}_{[i,i+k-1]}^{RC}$, then by Lemma \ref{lem:SSA}, we have $\bm{x} = \bm{y}_{[1,i+k-1]} \bm{y}_{[i+2k,n+k]}$. 
Since comparing two sequences of length $k$ requires $O(k)$ time and up to $n$ pairs may be checked in the worst case, this yields a decoding time of $O(kn)$. If the duplication length $k$ can be as large as $O(n)$, the worst-case time becomes $O(n^2)$. The following lemma develops a more efficient approach.

\begin{lemma}\label{lem:dedup}
    Let $\bm{x}\in \Sigma_q^n$ be an $m$-RCD root. If $\bm{x}$ suffers a reverse-complement duplication of length $k\geq 3m-3$ at position $i$ and results in $\bm{y}$, let $p\in [1,n-3m+4]$ be the smallest index such that $\bm{y}_{[p+3m-3,p+6m-7]}=\bm{y}_{[p,p+3m-4]}^{RC}$, then $i=p+3m-3-k$ and $\bm{x} = \bm{y}_{[1,p+3m-4]} \bm{y}_{[p+3m-3+k,n+k]}$.  
    This decoding process runs in $O(mn)$ time.
\end{lemma}

\begin{proof}
    Since $\bm{y}=RC_{k,i}(\bm{x})$, we have $\bm{y}_{[1,i+k-1]}=\bm{x}_{[1,i+k-1]}$ and $\bm{y}_{[i+k,i+2k-1]}=\bm{y}_{[i,i+k-1]}^{RC}$.
    Given that $k\geq 3m-3$, we can further derive $\bm{y}_{[i+k,i+k+3m-4]}=\bm{y}_{[i+k-3m-3,i+k-1]}^{RC}$.
    Consider the strings $\bm{y}_{[1,i+k-1]}$ and $\bm{y}_{[1,i+k+3m-4]}$, we have $\bm{y}_{[1,i+k+3m-4]}= RC_{3m-3,i+k-3m+3}(\bm{y}_{[1,i+k-1]})$.
    Since $\bm{x}$ is an $m$-RCD root, we see that the prefix $\bm{y}_{[1, i+k-1]}$ also maintains the $m$-RCD root property.
    By Lemma \ref{lem:SSA}, let $p\in [1,n-3m+4]$ be the smallest index such that $\bm{y}_{[p+3m-3,p+6m-7]}=\bm{y}_{[p,p+3m-4]}^{RC}$, we have $p=i+k-3m+3$. Therefore, we can determine $i=p+3m-3-k$ and then $\bm{x} = \bm{y}_{[1,p+3m-4]} \bm{y}_{[p+3m-3+k,n+k]}$.
    Since comparing two sequences of length $3m-3$ requires $O(m)$ time and up to $n$ pairs may be checked in the worst case, this yields a decoding time of $O(mn)$.
    This completes the proof.
\end{proof}

\begin{remark}
  When $m=2$, Lemma \ref{lem:dedup} suggests a decoding strategy for $m$-RCD roots that can correct a reverse-complement duplication of length $k$ with $k \ge 3m-3 = 3$. In fact, this strategy remains valid for $k=2$, as explained below.
  
  Let $\boldsymbol{x} \in \Sigma_q^n$ be a $2$-RCD root. Suppose $\boldsymbol{y} = RC_{2,i}(\boldsymbol{x})$ for some $i \in [1, n-1]$. Then $\boldsymbol{x} = \boldsymbol{y}_{[1,i+1]} \boldsymbol{y}_{[i+4,n+2]}$ and $y_{i+3} = \overline{y_i}$.
  Let $p \in [1, n-1]$ be the smallest index such that $\boldsymbol{y}_{[p+2,p+3]}=\boldsymbol{y}_{[p,p+1]}^{RC}$.
  We have $p \in \{i-1,i\}$; otherwise $\boldsymbol{x}_{[p+2,p+3]}=\boldsymbol{x}_{[p,p+1]}^{RC}$ for some $p \le i-2$, which contradicts the $2$-RCD root property.
\begin{itemize}
  \item If $p=i$, then $\boldsymbol{x} = \boldsymbol{y}_{[1,p+1]} \boldsymbol{y}_{[p+4,n+2]}$.
  \item If $p=i-1$, then $\boldsymbol{y}_{[i+1,i+2]}=\boldsymbol{y}_{[i-1,i]}^{RC}$. This implies $y_{i+1} = \overline{y_i}$. Since $y_{i+3} = \overline{y_i}$, we obtain $y_{i+3} = y_{i+1}$ and $\boldsymbol{y}_{[1,p+1]} \boldsymbol{y}_{[p+4,n+2]} = \boldsymbol{y}_{[1,i+1]} \boldsymbol{y}_{[i+4,n+2]}$.
\end{itemize}
In either case, the decoding strategy outputs the correct $\boldsymbol{x} = \boldsymbol{y}_{[1,i+1]} \boldsymbol{y}_{[i+4,n+2]}$. Therefore, the set of $2$-RCD roots can correct a reverse-complement duplication of length at least two, thereby recovering the main results established in \cite{Liu-25-AMC}. Moreover, compared with \cite{Liu-25-AMC}, we propose an efficient decoding strategy.
\end{remark}

\subsection{Extension to Correct Disjoint Duplications}

Now we generalize Lemma \ref{lem:dedup} to correct disjoint fixed-length reverse-complement duplications. We begin by introducing relevant definitions.

\begin{definition}
    Suppose $t,k$ and $i_1,\ldots,i_t$ are positive integers such that 
    \[
        1\leq i_1<\cdots<i_t\leq n-k+1 \quad \text{and} \quad i_j-i_{j-1}\geq k \quad \text{for } j\in [1,t].
    \]
    Let 
    \[
        \bm{x} = \bm{u}^{(1)} \bm{v}^{(1)} \bm{u}^{(2)} \bm{v}^{(2)} \cdots \bm{u}^{(t)} \bm{v}^{(t)} \bm{w}\in \Sigma_q^n,
    \]
    where $|\bm{w}|\geq 0$ and $|\bm{u}^{(j)}|\geq 0$, $|\bm{v}^{(j)}|=k$, and $|\bm{u}^{(j)}|+\sum_{p=1}^{j-1} (|\bm{u}^{(p)}|+|\bm{v}^{(p)}|)=i_j-1$ for $j\in [1,t]$.
    We say that $\bm{x}$ suffers \emph{$t$ disjoint $k$-reverse-complement duplications} at positions $i_1,\ldots,i_t$ if the resultant string is
    \[
    RC_{k,(i_1,i_2,\ldots,i_t)}(\bm{x})=\bm{u}^{(1)} \bm{v}^{(1)} \big(\bm{v}^{(1)}\big)^{RC} \bm{u}^{(2)} \bm{v}^{(2)} \big(\bm{v}^{(2)}\big)^{RC} \cdots \bm{u}^{(t)} \bm{v}^{(t)} \big(\bm{v}^{(t)}\big)^{RC} \bm{w}.
    \]
    A code $\mathcal{C} \subseteq \Sigma_q^n$ is termed a \emph{$t$-disjoint $k$-reverse-complement-duplications correcting code} if it can correct $t$ disjoint $k$-reverse-complement duplications.
\end{definition}

\begin{algorithm}
    \caption{Deduplication Algorithm}
    \label{alg:dedup}
    \begin{algorithmic}[1]
        \Require {$\bm{y} \in RC_{k,(i_1,i_2,\ldots,i_t)}(\bm{x})$, where $\bm{x}\in \Sigma_q^n$ is an $m$-RCD root, $k\geq 3m-3$, $t\geq 1$, $1\leq i_1<\cdots i_t\leq n-k+1$, and $i_j-i_{j-1}\geq k$ for $j\in [2,t]$}
        \Ensure {$\bm{x}$}

        \State \textbf{Initialization:} Let $p=1$ and $n'\triangleq|\bm{y}|=n+tk$

        \While{$t>0$}
            \If{$\bm{y}_{[p+3m-3,p+6m-7]}=\bm{y}_{[p,p+3m-4]}^{RC}$}
                \State $\bm{y}\leftarrow \bm{y}_{[1,p+3m-4]}\bm{y}_{[p+3m-3+k,n']}$
                \State $n'\leftarrow n'-k$
                \State $t\leftarrow t-1$
                \State $p\leftarrow p+k$
            \Else
                \State $p\leftarrow p+1$
            \EndIf
        \EndWhile
        \State $\bm{x}\leftarrow \bm{y}$
        \State \Return $\bm{x}$
    \end{algorithmic}
\end{algorithm}

\begin{theorem}\label{thm:rcd}
    Let $\mathcal{C}\subseteq \Sigma_q^n$ be the set of all $m$-RCD roots, then $\mathcal{C}$ is a $t$-disjoint $k$-reverse-complement-duplications correcting code for $t\geq 1$ and $k\geq 3m-3$. Moreover, Algorithm \ref{alg:dedup} is correct and runs in $O(mn)$ time.
\end{theorem}

\begin{proof}
    Let $\bm{x}\in \mathcal{C}$ and $\bm{y} \in RC_{k,(i_1,i_2,\ldots,i_t)}(\bm{x})$, where $1\leq i_1<\cdots i_t\leq n-k+1$ and $i_j-i_{j-1}\geq k$ for $j\in [2,t]$. We will show that Algorithm \ref{alg:dedup} outputs the correct $\bm{x}$.
    In the loop from Step $2$ to Step $10$, we first locate the smallest index $p_1$ such that $\bm{y}_{[p_1+3m-3,p_1+6m-7]}=\bm{y}_{[p_1,p_1+3m-4]}^{RC}$.
    By Lemma \ref{lem:dedup}, we have $i_1=p_1+3m-3-k$. 
    In Step $4$, we perform the deduplication operation by removing the inserted substring $\bm{y}_{[i_1+k,i_1+2k-1]}=\bm{y}_{[p_1+3m-3,p_1+3m-3+k-1]}$ and update $\bm{y}\leftarrow \bm{y}_{[1,p+3m-4]}\bm{y}_{[p+3m-3+k,n']}$. 
    Recall that $n'$ represents the length of $\bm{y}$, in Step $5$, we update $n'\leftarrow n'-k$.
    Since we have corrected the leftmost duplication error, Step $6$ decreases $t$ by one.
    Now, $i_2$ represents the position of the first duplication error and $i_2-i_1\geq k$, we seek the index $p_2=i_2+k-3m+3\geq i_1+2k-3m+3=p_1+k$ at which $\bm{y}_{[p_2+3m-3,p_2+6m-7]}=\bm{y}_{[p_2,p_2+3m-4]}^{RC}$, and thus in Step $7$ we skip $k$ positions.
    After this loop, the value of $t$ is reduced to zero, meaning that all duplication errors have been corrected. Therefore, the algorithm outputs the correct $\bm{x}$.

    In what follows, we analyze the complexity of this algorithm. Since $i_t\leq n-k+1$, in the last execution of the loop (Steps $2$–$10$) we have $p=i_t+k-3m+3\leq n$. Hence, the loop runs at most $n$ times.
    In each iteration, we compare two substrings of length $3m-3$.
    Consequently, the overall time complexity is $O(mn)$. This completes the proof.
\end{proof}

\subsection{RCD Roots}

To evaluate our code, we need to calculate the total number of $m$-RCD roots.

\begin{lemma}\label{lem:number}
  Let $A(n,m)$ be the number of $m$-RCD roots in $\Sigma_q^n$. If $m \geq \lceil \log_q n \rceil +1$, then $A(n,m) \geq (q-1)q^{n-1}$.
\end{lemma}

\begin{proof}
    We prove the lemma by a probabilistic analysis. 
    Choose a string $\bm{x}\in \Sigma_q^n$ uniformly at random, let $\mathbb{P}$ be the probability that $\bm{x}$ does not satisfy the $m$-RCD root property, i.e., there exists some $i\in [1,n-2m+1]$ such that $\bm{x}_{[i+m,i+2m-1]}=\bm{x}_{[i,i+m-1]}^{RC}$. We will show that $\mathbb{P}\leq \frac{1}{q}$.
    
    For $i\in [1,n-2m+1]$, the probability that $\bm{x}_{[i+m,i+2m-1]}=\bm{x}_{[i,i+m-1]}^{RC}$ is $\frac{1}{q^m}\leq \frac{1}{nq}$. By the union bound, we can compute $\mathbb{P}\leq (n-2m+1)\cdot \frac{1}{q^m} \leq \frac{1}{q}$. Then, we get $A(n,m) \geq q^n (1- \frac{1}{q})= (q-1)q^{n-1}$. This completes the proof.
\end{proof}

A natural question concerns how to encode a string into an $m$-RCD root efficiently. In what follows, we adopt the general framework presented in \cite[Construction 2]{Kobovich-24-ISIT}
to encode a string into an $m$-RCD root with one redundant symbol, under the condition that $m\geq \lceil \log_q n\rceil +1$.

\begin{algorithm}
    \caption{Universal Iterative Encoder $ENC$}
    \label{alg:enc}
    \begin{algorithmic}[1]
        \Require {$\bm{x} \in \Sigma_q^{n-1}$}
        \Ensure {$\bm{y}=ENC(\bm{x})\in \mathcal{C}\subseteq \Sigma_q^{n}$}

        \State $\bm{y}\leftarrow \bm{x}1$

        \While{$\mathbbm{1}_{\mathcal{C}}(\bm{y})=0$}
            \State $\bm{y}\leftarrow \xi(\bm{y})0$
        \EndWhile
        \State \Return $\bm{y}$
    \end{algorithmic}
\end{algorithm}

\begin{algorithm}
    \caption{Universal Iterative Decoder $DEC$}
    \label{alg:dec}
    \begin{algorithmic}[1]
        \Require {$\bm{y}=ENC(\bm{x}) \in \mathcal{C}$ for some $\bm{x}\in \Sigma_q^{n-1}$}
        \Ensure {$\bm{x}=DEC(\bm{y})$}

        \While{$y_n=0$}
            \State $\bm{y}\leftarrow \xi^{-1}(\bm{y}_{[1,n-1]})$
            {\color{gray}\Comment{$\xi^{-1}$ represents the `inverse' of $\xi$}}
        \EndWhile
        \State $\bm{x}\leftarrow \bm{y}_{[1,n-1]}$
        \State \Return $\bm{x}$
    \end{algorithmic}
\end{algorithm}

\begin{lemma}\cite[Construction 2]{Kobovich-24-ISIT}\label{Kobovich}
    Let $\mathcal{C}\subseteq \Sigma_q^n$ be a given constraint such that $|\mathcal{C}|\geq q^{n-1}$. Given an injective function $\xi: \Sigma_q^n\setminus \mathcal{C}\rightarrow \Sigma_q^{n-1}$ and an indicator function $\mathbbm{1}_{\mathcal{C}}: \Sigma_q^n\rightarrow \{0,1\}$, then Algorithms \ref{alg:enc} and \ref{alg:dec} present an efficient channel construction with one redundant symbol and $O\big(T(n)\big)$ average time complexity, where $T(n)$ denotes the maximal time complexity among $\xi,\xi^{-1},\mathbbm{1}_{\mathcal{C}}$.
\end{lemma}

\begin{theorem}\label{thm:enc}
    Set $m=\lceil \log_q n\rceil+1$.
    Let $\mathcal{C}\subseteq \Sigma_q^n$ be the set of $m$-RCD roots. 
    We define the functions $\xi:  \Sigma_q^n\setminus \mathcal{C}\rightarrow \Sigma_q^{n-1}$, $\xi^{-1}: \Sigma_q^{n-1}\rightarrow \Sigma_q^n$, and $\mathbbm{1}_{\mathcal{C}}: \Sigma_q^n\rightarrow \{0,1\}$ as follows:
    \begin{itemize}
        \item For any $\bm{x}\notin \mathcal{C}$, $\xi(\bm{x})=\bm{x}_{[1,i+m-1]} \bm{x}_{[i+2m,n]}Rep_{q,m-1}(i)$, where $i$ represents the smallest index such that $\bm{x}_{[i+m,i+2m-1]}=\bm{x}_{[i,i+m-1]}^{RC}$ and $Rep_{q,m-1}(i)$ denotes the $q$-ary representation of $i$ of length $m-1=\lceil \log_q n\rceil$.
        \item For any $\bm{y}=\xi(\bm{x})$ with $\bm{x}\notin \mathcal{C}$, $\xi^{-1}(\bm{y})= \bm{y}_{[1,i+m-1]} \bm{y}_{[i,i+m-1]}^{RC} \bm{y}_{[i+m,n-m]}$, where $i$ denotes the integer representation of $\bm{y}_{[n-m+1,n-1]}$.
        \item For any $\bm{x}\in \Sigma_q^n$, $\mathbbm{1}_{\mathcal{C}}(\bm{x})=0$ if $\bm{x}\notin \mathcal{C}$ and $\mathbbm{1}_{\mathcal{C}}(\bm{x})=1$ otherwise.
    \end{itemize}
    Then Algorithms \ref{alg:enc} and \ref{alg:dec} present an efficient channel construction for $m$-RCD roots with one redundant symbol and $O(mn)$ average time complexity.
\end{theorem}

\begin{proof}
    By Lemmas \ref{lem:number} and \ref{Kobovich}, it suffices to show that $\xi^{-1}$ is the `inverse' of $\xi$, i.e., $\xi^{-1}\big(\xi(\bm{x})\big)=\bm{x}$ for $\bm{x}\notin \mathcal{C}$, and that the maximal time complexity among $\xi,\xi^{-1},\mathbbm{1}_{\mathcal{C}}$ is $O(mn)$.

    For any $\bm{x}\notin \mathcal{C}$, let $i$ be the smallest index such that $\bm{x}_{[i+m,i+2m-1]}=\bm{x}_{[i,i+m-1]}^{RC}$ and $Rep_{q,m-1}(i)$ be the $q$-ary representation of $i$ of length $m-1$, then $\xi(\bm{x})=\bm{x}_{[1,i+m-1]} \bm{x}_{[i+2m,n]}Rep_{q,m-1}(i)$.
    It follows that 
    \begin{align*}
        \xi^{-1}\big(\xi(\bm{x})\big)
        &= \bm{x}_{[1,i+m-1]} \bm{x}_{[i,i+m-1]}^{RC} \bm{x}_{[i+2m,n]}\\
        &= \bm{x}_{[1,i+m-1]} \bm{x}_{[i+m,i+2m-1]} \bm{x}_{[i+2m,n]}\\
        &=\bm{x}.
    \end{align*}
    Thus, $\xi^{-1}$ is the `inverse' of $\xi$.

    Since finding the smallest index $i<n$ such that $\bm{x}_{[i+m,i+2m-1]}=\bm{x}_{[i,i+m-1]}^{RC}$ requires $O(mn)$ time and obtaining the $q$-ary representation $Rep_{q,m-1}(i)$ of $i$ (as well as the integer representation of $Rep_{q,m-1}(i)$) requires $O(n)$ time, the maximal time complexity among $\xi,\xi^{-1},\mathbbm{1}_{\mathcal{C}}$ is $O(mn)$. This completes the proof.
\end{proof}

By Theorems \ref{thm:rcd} and \ref{thm:enc}, we can derive the following conclusion.

\begin{corollary}\label{cor:summary}
  There exists a code $\mathcal{C}\subseteq \Sigma_q^n$ with one redundant symbol capable of correcting an arbitrary number of disjoint $k$-reverse-complement duplications, provided that $k \geq 3\lceil \log_q n \rceil$. Both the encoding and decoding processes for this code can be efficiently completed in $O(n\log_q n)$ average time.
\end{corollary}

\begin{remark}
  Recall that the complement operation can be viewed as a permutation $\pi:\Sigma_q \to \Sigma_q$ with $\pi(a)\neq a$ and $\pi\big(\pi(a)\big)=a$ for each $a\in\Sigma_q$. Since all results established herein do not require the condition $\pi(a)\neq a$, these conclusions remain valid when replacing the complement operation with the identity permutation. Consequently, all results established herein can be seamlessly extended to handle palindromic duplications by simply replacing instances of `reverse-complement' with `palindromic'.
\end{remark}

\section{Correcting a Palindromic/Reverse-Complement Duplication of Arbitrary Length}\label{sec:arbitrary}

Liu, Tang, Fan, and Sagar \cite{Liu-25-AMC} and Lenz, Wachter-Zeh, and Yaakobi \cite{Lenz-19-DCC} showed that there exist codes capable of correcting a reverse-complement duplication and a palindromic duplication, respectively, of arbitrary length, with rate $\log_q (q-1)$.
In this section, we derive a Gilbert-Varshamov bound for such codes, showing that the optimal redundancy is at most $2\log_q n+\log_q\log_q n+O(1)$, i.e., the asymptotic optimal rate is $1$.
In what follows, we focus solely on reverse-complement duplications, since all results established herein can be seamlessly
extended to handle palindromic duplication by simply replacing instances of `reverse-complement' with `palindromic'.
Before proceeding, we introduce some essential graph-theoretic tools.

\begin{definition}
  Let $ G = (V, E) $ be a graph, where $ V $ is the vertex set and $ E $ is the edge set. For a vertex $ \bm{v} \in V $, let $ d(\bm{v}) $ denote its degree, i.e., the number of vertices adjacent to $ \bm{v} $. An \emph{independent set} in $ G $ is a subset of $ V $ such that no two distinct vertices are connected by an edge. The \emph{independence number} of $ G $, denoted by $ \alpha(G) $, is the cardinality of the largest independent set in $ G $.
\end{definition}

\begin{lemma}\cite[Page 100, Theorem 1]{Alon}\label{lem:Alon}
  For any graph $ G = (V, E) $, the independence number satisfies 
  \[
  \alpha(G) \geq \sum_{\bm{v} \in V} \frac{1}{d(\bm{v}) + 1}.
  \]
\end{lemma}

We now consider the graph $G^* = (V, E)$, where $ V $ is a subset of $ \Sigma_q^n $. Two distinct vertices $ \bm{x}, \bm{z} \in V $ are adjacent, i.e., $ (\bm{x}, \bm{z}) \in E $, if and only if there exists some $ k \geq 1 $ such that $ RC_k^1(\bm{x}) \cap RC_k^1(\bm{z}) \neq \varnothing $. In this context, a code $ \mathcal{C} \subseteq V $ can correct a reverse-complement duplication of arbitrary length if and only if $ \mathcal{C} $ forms an independent set in $ G^* $. Consequently, the maximum size of such a code $ \mathcal{C} $ is equal to $ \alpha(G^*) $.

To establish a lower bound on $ \alpha(G^*) $, we apply Lemma~\ref{lem:Alon}. This requires deriving an upper bound on the degree $ d(\bm{x}) $ for each vertex $ \bm{x} \in V $. To this end, we define
\[
N_{k}(\bm{x}; V) = \left\{ \bm{z} \in V: \bm{z} \neq \bm{x} \text{ and } RC_k^1(\bm{z}) \cap RC_k^1(\bm{x}) \neq \varnothing \right\}.
\]

\begin{lemma}\label{lem:N}
  For any $ V \subseteq \Sigma_q^n $ and $ \bm{x} \in V $, it holds that
  \[
  \#N_{k}(\bm{x}; V) \leq n(n - k + 1) - 1.
  \]
\end{lemma}

\begin{proof}
  Each vertex $ \bm{z} \in N_{k}(\bm{x}; V) \cup \{ \bm{x} \} $ can be derived from $ \bm{x} $ through two operations:
  \begin{enumerate}
    \item Perform a reverse-complement duplication of length $ k $ to obtain a string $ \bm{y} $ of length $ n + k $. The number of ways to choose the starting position of the duplication is at most $ n - k + 1 $, so the number of such $ \bm{y} $ is upper bounded by $ n - k + 1 $.
    \item Perform a reverse-complement deduplication of length $ k $ to obtain a string $ \bm{z} $ of length $ n $. The number of ways to choose the starting position of the deduplication is at most $ n $, so the number of such $ \bm{z} $ is upper bounded by $ n(n - k + 1) $.
  \end{enumerate}
  Therefore, $ \#N_{k}(\bm{x}; V) \leq n(n - k + 1) - 1 $. This completes the proof.
\end{proof}

Using Lemma~\ref{lem:N}, we compute
\[
d(\bm{x}) \leq \sum_{k=1}^n \#N_{k}(\bm{x}; V) \leq \sum_{k=1}^n \big( n(n - k + 1) - 1 \big) \leq n^3 - 1.
\]
By Lemma~\ref{lem:Alon}, this gives
\[
\alpha(G^*) \geq \frac{\#V}{n^3}.
\]
Setting $ V = \Sigma_q^n $ (i.e., all possible strings), we obtain the following result.

\begin{lemma}
  There exists a code capable of correcting a reverse-complement duplication of arbitrary length with at most $ 3 \log_q n $ redundant symbols.
\end{lemma}

The following theorem provides an improved bound by selecting $ V $ as the set of all $ m $-RCD roots.

\begin{theorem}\label{thm:GV}
  There exists a code capable of correcting a reverse-complement duplication of arbitrary length with at most $ 2\log_q n + \log_q\log_q n + O(1) $ redundant symbols, i.e., with asymptotic rate $1$.
\end{theorem}

\begin{proof}
  Let $ G^* = (V, E) $ be defined above. Let $ V $ be the set of $ m $-RCD roots with $ m = \lceil \log_q n \rceil + 1 $. By Lemma~\ref{lem:SSA}, we have $ \#V \geq (q - 1) q^{n - 1} $. Furthermore, by Theorem~\ref{thm:rcd}, we have $ \#N_{k}(\bm{x}; V) = 0 $ for $ k \geq 3\lceil \log_q n \rceil $.
  As a result, by Lemma~\ref{lem:N}, we can compute
  \[
  d(\boldsymbol{x}) \leq \sum_{k=1}^{3\lceil \log_q n \rceil - 1} \#N_{k}(\bm{x}; V) \leq 3n^2\lceil \log_q n \rceil - 1.
  \]
  Finally, by Lemma~\ref{lem:Alon}, we obtain
  \[
  \alpha(G^*) \geq \frac{(q - 1) q^{n - 1}}{3n^2\lceil \log_q n \rceil} = O\left(\frac{q^n}{n^2\log_q n}\right).
  \]
  Then the conclusion follows.
\end{proof}

\begin{remark}
  In this remark, we outline a potential approach to construct codes capable of correcting a reverse-complement duplication of arbitrary length with redundancy superior to the Gilbert-Varshamov bound derived in Theorem~\ref{thm:GV}.

  Building upon the work of Bitar, Hanna, Polyanskii, and Vorobyev \cite{Bitar-21-ISIT}, who constructed a binary code $ \mathcal{C}_1 \subseteq \Sigma_2^n $ capable of correcting a burst-insertion of length at most $ k = O\left( \frac{n}{(\log_2 n)^2} \right) $ with redundancy $ \log_2 n + O\left( k (\log_2 (k\log_2 n))^2 \right) $, we propose the following construction.
  
  Let $ \mathcal{C}_2 \subseteq \Sigma_2^n $ be the set of $ m $-RCD roots with $ m \leq \frac{k + 4}{3} $, by Theorem~\ref{thm:rcd}, $\mathcal{C}_2$ can correct a reverse-complement duplication of length at least $k+1$. 
  Then the code $ \mathcal{C} = \mathcal{C}_1 \cap \mathcal{C}_2 $ can correct a reverse-complement duplication of arbitrary length. By judiciously selecting the parameter $k$, it is possible for the redundancy of $ \mathcal{C} $ to be at most $ \log_2 n + o(\log_2 n) $. 
  This construction can be generalized to $q$-ary alphabets, as $q$-ary symbols can be viewed as binary strings of length $\lceil \log_2 q \rceil$. 
\end{remark}

\section{Correcting Length-One Reverse-Complement Duplications}\label{sec:short}

In this section, we study codes for correcting short reverse-complement duplications.
Ben-Tolila and Schwartz \cite{Ben-Tolila-22-IT} showed that when the duplication length $k$ is odd, any binary code capable of correcting a burst-insertion of length $k$ can be utilized through the complement preserving mapping to construct $q$-ary codes that correct a reverse-complement duplication of length $k$. 

\begin{definition}
  Let $\beta: \Sigma_q\rightarrow \Sigma_2$ be a mapping. We say that $\beta$ is \emph{complement-preserving} if $\beta(\overline{a})= \overline{\beta(a)}$ for $a\in \Sigma_q$, i.e., $\beta$ and the complement operation commute.
  This complement-preserving mapping extends naturally to entire strings, i.e., for any $\bm{x}=x_1x_2\cdots x_n \in \Sigma_q^n$, we define $\bm{\beta}(\bm{x})= \beta(x_1) \beta(x_2)\cdots \beta(x_n)$.
\end{definition}

\begin{lemma}\cite[Theorem 17]{Ben-Tolila-22-IT}
    Let $\beta: \Sigma_q\rightarrow \Sigma_2$ be a complement-preserving mapping and $k$ be an odd integer. If $\mathcal{C}'\subseteq \Sigma_2^n$ can correct a burst-insertion of length $k$, then the code $\mathcal{C}\triangleq \big\{\boldsymbol{x}\in \Sigma_q^n:\bm{\beta}(\boldsymbol{x})\in \mathcal{C}' \big\}$ can correct a reverse-complement duplication of length $k$.
\end{lemma}

A natural question is whether this conclusion  extends to even $k$ and to scenarios with multiple reverse-complement duplications. Unfortunately, the answer is negative.

\begin{example}
  Let $n=6$, $q=4$, $\overline{0}\triangleq 1$, and $\overline{2}=3$. Let $\beta: \Sigma_q \rightarrow  \Sigma_2$ be a complement-preserving mapping such that $\beta(0)=\beta(2)=0$ and $\beta(1)=\beta(3)=1$.
  Define $\mathcal{C}'=\{010101\}\subseteq \Sigma_2^n$ and $\mathcal{C}=\{232301,230101\}\subseteq \Sigma_q^n$. Since $\mathcal{C}'$ contains a single codeword, it trivially corrects $t$ $k$-reverse-complement-duplications for $t, k\geq 1$. Moreover, since $\bm{\beta}(232301)=\bm{\beta}(230101)=010101$, we have $\mathcal{C}\subseteq \big\{\boldsymbol{x}\in \Sigma_q^n:\bm{\beta}(\boldsymbol{x})\in \mathcal{C}' \big\}$. We will show that $\mathcal{C}$ cannot correct a length-two reverse-complement duplication, nor two length-one reverse-complement duplications. 
  \begin{itemize}
    \item If $k=2$ and $t=1$, we have $2323\underline{01}\rightarrow 23230101$ and $\underline{23}0101\rightarrow 23230101$, where the underlined substring denotes the duplicated portion. This implies that both $232301$ and $230101$ can evolve into $23230101$ via a length-two reverse-complement duplication.
    Therefore, $\mathcal{C}$ can not correct a length-two reverse-complement duplication.
    \item If $k=1$ and $t=2$, we have $23230\underline{1}\rightarrow 232301\underline{0} \rightarrow 23230101$ and $2\underline{3}0101\rightarrow 23\underline{2}0101\rightarrow 23230101$. This implies that both $232301$ and $230101$ can evolve into $23230101$ via two length-one reverse-complement duplications. Therefore, $\mathcal{C}$ can not correct two length-one reverse-complement duplications.
  \end{itemize}
\end{example}

This motivates developing new tools to handle a reverse-complement duplication of even length or multiple reverse-complement duplications. In what follows, we focus on multiple reverse-complement duplications with the restriction that each duplication length is one.
We begin by examining the similarities and distinctions between  a length-one reverse-complement duplication and an insertion.
Let $\bm{x}\in \Sigma_q^n$.
Recall that $RC_1^1(\bm{x})$ denotes the $1$-reverse-complement ball centered at $\bm{x}$ of radius one. Let $I_1(\boldsymbol{x})$ denote the $1$-insertion ball of $\boldsymbol{x}$, i.e., the set of strings obtainable from $\boldsymbol{x}$ by exactly one insertion.

\begin{lemma}\label{lem:distinct}
  For any $\bm{x}\in \Sigma_q^n$, it holds that 
  \begin{gather*}
    RC_1^1(\bm{x})=\{\bm{x}_{[1,i]}\overline{x}_i\bm{x}_{[i+1,n]}:i\in [1,n]\},\\
    I_1(\bm{x})=\{a\bm{x}:a\in \Sigma_q\} \cup \{\bm{x}_{[1,i]}a\bm{x}_{[i+1,n]}:i\in [1,n],a\in \Sigma_q,a\neq x_i\}.
  \end{gather*}
\end{lemma}

\begin{proof}
    By definition, we have $RC_1^1(\bm{x})=\{\bm{x}_{[1,i]}\overline{x}_i\bm{x}_{[i+1,n]}:i\in [1,n]\}$ and $I_1(\bm{x})=\{a\bm{x}:a\in \Sigma_q\} \cup \{\bm{x}_{[1,i]}a\bm{x}_{[i+1,n]}:i\in [1,n],a\in \Sigma_q\}$.
    Consider the sequence $\bm{x}_{[1,i]}x_i\bm{x}_{[i+1,n]}$ for some $i\in [1,n]$, 
    \begin{itemize}
        \item if $x_j=x_i$ for $j\in [1,i]$, then $\bm{x}_{[1,i]}x_i\bm{x}_{[i+1,n]}=x_i\bm{x}\in \{a\bm{x}:a\in \Sigma_q\}$;
        \item if $x_j\neq x_i$ for some $j\in [1,i]$, let $p\in [1,i-1]$ be the largest index such that $x_p\neq x_i$, then $\bm{x}_{[1,i]}x_i\bm{x}_{[i+1,n]}=\bm{x}_{[1,p]}x_i\bm{x}_{[p+1,n]}\in \{\bm{x}_{[1,p]}a\bm{x}_{[p+1,n]}:a\in \Sigma_q,a\neq x_p\}$.
    \end{itemize}
    Therefore, we have $I_1(\bm{x})=\{a\bm{x}:a\in \Sigma_q\} \cup \{\bm{x}_{[1,i]}a\bm{x}_{[i+1,n]}:i\in [1,n],a\in \Sigma_q,a\neq x_i\}$. This completes the proof.
\end{proof}

Lemma \ref{lem:distinct} suggests that, in the binary case, there is virtually no difference between length-one reverse-complement duplications and insertions, whereas the difference becomes significant in the non-binary case. Therefore, it is difficult to develop binary codes that outperform general insertion-correcting codes \cite{Guruswami-21-IT,Levenshtein-66,Li-23-ISIT,Sima-20-ISIT,Song-22-IT,Sun-24-IT} for correcting length-one reverse-complement duplications, and we move on to non-binary alphabets.
In the rest of this section, we assume that $q\geq 4$ is a fixed even integer and define the complement of $a\in \Sigma_q^n$ as 
\begin{equation*}
    \overline{a}=
    \begin{cases}
        2\lfloor a/2 \rfloor+1, &\mbox{if } 2|a;\\
        2\lfloor a/2 \rfloor, &\mbox{if } 2\nmid a.
    \end{cases}
\end{equation*}
In other words, let $i=\lfloor a/2 \rfloor$, then $\{a,\overline{a}\}=\{2i,2i+1\}$.
Clearly, this complement operation is well-defined.

\begin{definition}
    For any $a\in \Sigma_q$, let 
    \[
        a^{\oplus}= a\{a,\overline{a}\}^{\ast}
    \]
    denote the set of all strings that start with $a$ and are followed by any finite length string composed of $a$ and $\overline{a}$ only. In a string $\bm{x}\in \Sigma_q^n$, a substring $\bm{x}_{[i,j]}$, where $1\leq i\leq j\leq n$, is called a \emph{run}\footnote{In the traditional definition (e.g., as in \cite{Schoeny-17-IT}), a run is a substring consisting of a single symbol with maximal length.} if $\bm{x}_{[i,j]}\in x_i^{\oplus}$ and it has maximal length, i.e., $x_{i-1}\not\in \{x_i,\overline{x_i}\}$ (when $i>1$) and $x_{j+1}\not\in \{x_i,\overline{x_i}\}$ (when $j<n$). 
\end{definition}

\begin{definition}
    For any $\bm{x}\in \Sigma_q^n$, let $r(\bm{x})$ be the number of runs in $\bm{x}$. For $i\in [1,r(\bm{x})]$, let $\bm{\sigma}(\bm{x},i)$ be the $i$-th run of $\bm{x}$ and $\sigma(\bm{x})_i=\sigma(\bm{x},i)_1$ be the first entry of $\bm{\sigma}(\bm{x},i)$, we then define the \emph{signature} of $\bm{x}$ as 
    \[
        \bm{\sigma}(\bm{x})=\big(\sigma(\bm{x})_1, \sigma(\bm{x})_2, \cdots, \sigma(\bm{x})_{r(\bm{x})}\big).
    \]
\end{definition}

\begin{example}\label{ex:signature}
    Let $q=4$, then $\overline{0}=1$ and $\overline{2}=3$. Consider $\bm{x}=01123221001$, then 
    \begin{itemize}
        \item $r(\bm{x})=3$;
        \item $\bm{\sigma}(\bm{x},1)=011$, $\bm{\sigma}(\bm{x},2)=2322$, and $\bm{\sigma}(\bm{x},3)=1001$;
        \item $\sigma(\bm{x})_1=0$, $\sigma(\bm{x})_2=2$, and $\sigma(\bm{x})_3=1$;
        \item $\bm{\sigma}(\bm{x})=(0,2,1)$.
    \end{itemize}
\end{example}

The lemma below is a variation of \cite[Lemma 4]{Yohananov-25-DCC}. We provide a complete proof here for completeness.

\begin{lemma}\label{lem:transform}
    Let $t\geq 1$ be an integer. For any $\bm{x}\in \Sigma_q^n$ and $\bm{y}\in RC_1^t(\bm{x})$, then
    \begin{itemize}
        \item $\bm{\sigma}(\bm{x})=\bm{\sigma}(\bm{y})$, i.e., $r(\bm{x})=r(\bm{y})$ and  $\sigma(\bm{x})_i=\sigma(\bm{y})_i$ for $i\in [1,r(\bm{x})]$;
        \item there exist some non-negative integers $t_1,t_2,\ldots,t_{r(\bm{x})}$ with $t_1+t_2+\cdots+t_{r(\bm{x})}=t$ such that $\bm{\sigma}(\bm{y},i)\in RC_1^{t_i}\big(\bm{\sigma}(\bm{x},i)\big)$ for $i\in [1,r(\bm{x})]$.
    \end{itemize}
\end{lemma}

\begin{proof}
    Let $\bm{x}=\bm{y}^{(0)}$ and $\bm{y}^{(t)}=\bm{y}$.
    Since $\bm{y}\in RC_1^t(\bm{x})$, there exist $t-1$ intermediate sequences $\bm{y}^{(1)},\ldots, \bm{y}^{(t-1)}$ such that $\bm{y}^{(i)}\in RC_1^1\big(\bm{y}^{(i-1)}\big)$ for $i\in [1,t]$. 
    Assume $\bm{y}^{(i)}$ is obtained from $\bm{y}^{(i-1)}$ via a duplication at a position within the $i_j$-th run of $\bm{y}^{(i-1)}$, then the effect of the duplication is simply to extend the $i_j$-th run of $\bm{y}^{(i-1)}$ by one symbol, while leaving its first symbol unchanged and not affecting the other runs. This implies that
    \[
        \bm{\sigma}\big(\bm{y}^{(i)}\big)=\bm{\sigma}\big(\bm{y}^{(i-1)}\big), \quad \bm{\sigma}\big(\bm{y}^{(i)},i_j \big)\in RC_1^1\big(\bm{\sigma}(\bm{y}^{(i-1)},i_j) \big), \quad \text{and} \quad \bm{\sigma}\big(\bm{y}^{(i)},s \big)= \bm{\sigma}\big(\bm{y}^{(i-1)},s \big) \quad \text{for } s\neq i_j.
    \]
    Then the conclusion follows.
\end{proof}

If we record each run as a special symbol and apply Lemma \ref{lem:transform}, length-one reverse-complement duplications can be viewed as substitutions. As a result, we can leverage the tools from the Hamming metric to correct reverse-complement duplications. 

\subsection{Runlength-Limited Constraint}

Since some runs may have length $O(n)$, treating each run as a single symbol yields a vector defined over a large alphabet, which would require substantial redundancy to correct substitutions within it. To reduce redundancy, it is therefore desirable to bound the maximum run length so that the vector is defined over a smaller alphabet. To this end, we introduce the notion of runlength-limited strings.

\begin{definition}
    A string $\bm{x}\in \Sigma_q^n$ is termed an \emph{$m$-runlength-limited string} (or \emph{$m$-RLL string} for short) if each run of $\bm{x}$ has length at most $m$, i.e.,
    \[
        |\bm{\sigma}(\bm{x},i)|\leq m \quad \text{for } i\in [1,r(\bm{x})].
    \]  
\end{definition}

\begin{lemma}\label{lem:rll}
  Let $B(n,m)$ be the number of $m$-RLL strings in $\Sigma_q^n$. If $q\geq 4$ is even and $m \geq \lceil \log_{q/2} n \rceil+1$, then $B(n,m) \geq q^n-2q^{n-1}\geq \frac{q^n}{2}$.
\end{lemma}

\begin{proof}
    We prove the lemma by a probabilistic analysis. 
    Choose a string $\bm{x}\in \Sigma_q^n$ uniformly at random, let $\mathbb{P}$ be the probability that $\bm{x}$ does not satisfy the $m$-RLL constraint, i.e., there exists some $i\in [1,n-m]$ such that $\bm{x}_{[i,i+m]}$ is a run. We will show that $\mathbb{P}\leq \frac{1}{2}$.
    
    For $i\in [1,n-m]$, the probability that $\bm{x}_{[i,i+m]}$ is a run can be calculated as $\frac{q\cdot 2^m}{q^{m+1}}= \frac{1}{(q/2)^m}\leq \frac{2}{qn}$. By the union bound, we can compute $\mathbb{P}\leq (n-m)\cdot \frac{2}{qn} \leq \frac{2}{q}$. Then, we get $B(n,m) \geq q^n (1- \frac{2}{q})= q^n-2q^{n-1}$. This completes the proof.
\end{proof}

In what follows, we use the sequence replacement technique to design an efficient encoder that maps a given string to an $m$-RLL string with one redundant symbol, under the condition that $m= \lceil \log_q n\rceil+ \big\lceil \lceil \log_q n+1\rceil\cdot \log_{q/2} 2 \big\rceil+1$. It should be noted that when $\log_q n$ and $(\log_q n+1)\cdot \log_{q/2} 2$ are integers, we have 
\begin{equation*}
\begin{aligned}
    \lceil \log_q n\rceil+ \big\lceil \lceil \log_q n+1\rceil\cdot \log_{q/2} 2 \big\rceil+1
    &= (1+\log_{q/2} 2)\cdot \log_q n+ \log_{q/2} 2+1\\
    &= \log_{q/2} n+ \log_{q/2} 2+1.
\end{aligned}
\end{equation*}
Therefore, this choice of $m$ is close to the lower bound derived in Lemma \ref{lem:rll}. The justification for this choice of $m$ is as follows: when there exists a run of length exceeding $m$, we first remove a substring of length $m$ inside it. In order to preserve the overall length and to make this process reversible, we then record the position of the removed substring using $\lceil \log_q n \rceil$ symbols, record the information of the removed substring using $\left\lceil \left\lceil \log_q n + 1 \right\rceil \cdot \log_{q/2} 2 \right\rceil$ symbols (the validity of which is established below), and use a single marker symbol to indicate whether a removal exists.

\begin{lemma}\label{lem:m}
  Let $m=\lceil \log_q n\rceil+m'+1$. It holds that $q^{m'}\geq 2^m$ if and only if $m'\geq \lceil \log_q n+1\rceil\cdot \log_{q/2} 2$.
\end{lemma}

\begin{proof}
  Observe that 
  \begin{align*}
    q^{m'}\geq 2^m 
    &\Leftrightarrow (q/2)^{m'}\geq 2^{\lceil \log_q n\rceil+1}\\
    &\Leftrightarrow m'\cdot \log_2 (q/2)\geq \lceil \log_q n\rceil+1\\
    &\Leftrightarrow m'\geq \lceil \log_q n+1\rceil\cdot \log_{q/2} 2.
  \end{align*}
  This completes the proof.
\end{proof}

The following functions will be used in the design of the RLL encoder.
\begin{definition}\label{def:h}
  For any $a\in \Sigma_q$ with even $q\geq 4$, we define
  \begin{gather*}
    h_1(a)=2\lfloor a/2 \rfloor +3 \pmod{q},\\
    h_2(a)=2\lfloor a/2 \rfloor +2 \pmod{q}.
  \end{gather*}
\end{definition}

\begin{lemma}\label{lem:h}
  Let $h_1(\cdot)$ and $h_2(\cdot)$ be defined above. For any $a\in \Sigma_q$ with even $q\geq 4$, $h_1(a)\notin\{a,\overline{a}\}$ is odd and $h_2(a)\notin\{a,\overline{a}\}$ is even.
\end{lemma}

\begin{proof}
  Let $i=\lfloor a/2 \rfloor$, then $\{a,\overline{a}\}=\{2i,2i+1\}$. It follows that $h_1(a)=2i+3 \pmod{q}\notin \{2i,2i+1\}$ is odd and $h_2(a)=2i+2 \pmod{q}\notin \{2i,2i+1\}$ is even. This completes the proof.
\end{proof}

\begin{definition}\label{def:indicator}
  For any run string $a_1a_2\cdots a_m\in a_1^{\oplus}$ with $a_1\in \Sigma_q$, we define its \emph{indicator vector} as 
  \[
    \mathbbm{1}(a_1a_2\cdots a_m)= \big(\mathbbm{1}(a_1),\mathbbm{1}(a_2),\ldots, \mathbbm{1}(a_m)\big),
  \]
  where $\mathbbm{1}(a_i)=1$ if $a_i=a_1$ and $\mathbbm{1}(a_i)=0$ otherwise, for $i\in [1,m]$.
  Moreover, we define the \emph{integer representation} of $\mathbbm{1}(a_1a_2\cdots a_m)$ as
  \[
    \varphi\big(a_1a_2\cdots a_m\big)=\sum_{j=1}^{m}2^{j-1} \cdot \mathbbm{1}(a_{m+1-j}).
  \] 
\end{definition}

  Note that $\mathbbm{1}(a_1)=1$ and $\varphi(a_1a_2\cdots a_m)\in [2^{m-1},2^m-1]$. As a result, the length of the run $a_1a_2\cdots a_m$ is determined by $\varphi(a_1a_2\cdots a_m)$.
  Consequently, from $\varphi(a_1a_2\cdots a_m)$, we can recover $\mathbbm{1}(a_1a_2\cdots a_m)$.
  If $a_1$ is known a priori,\footnote{This is a natural assumption, as the length-one reverse-complement duplication does not change the starting symbol of a run.} we can further recover $a_1a_2\cdots a_m$.
  This implies that when $a_1$ is known a priori, the inverse mapping $\varphi^{-1}$ defined by 
  \begin{equation}\label{eq:inverse}
    \varphi^{-1}\big(\varphi(a_1a_2\cdots a_m),a_1\big)= a_1a_2\cdots a_m
  \end{equation}
  is well-defined.
  
\begin{example}\label{ex:associate}
    Following Example \ref{ex:signature}, we consider the string $\bm{x}=01123221001$, then 
    \begin{itemize}
        \item $\mathbbm{1}\big(\bm{\sigma}(\bm{x},1)\big)=100$, $\mathbbm{1}\big(\bm{\sigma}(\bm{x},2)\big)=1011$, and $\mathbbm{1}\big(\bm{\sigma}(\bm{x},3)\big)=1001$;
        \item $\varphi\big(\bm{\sigma}(\bm{x},1)\big)=4$, $\varphi\big(\bm{\sigma}(\bm{x},2)\big)=11$, and $\varphi\big(\bm{\sigma}(\bm{x},3)\big)=9$.
    \end{itemize}
\end{example}

\begin{algorithm}
    \caption{$m$-RLL Encoder for $m= m_1+ m_2+1$, where $m_1=\lceil \log_q n\rceil$ and $m_2=\big\lceil \lceil \log_q n+1\rceil\cdot \log_{q/2} 2 \big\rceil$}
    \label{alg:rll_enc}
    \begin{algorithmic}[1]
        \Require {$\bm{x} \in \Sigma_q^{n-1}$}
        \Ensure {$\bm{y}=ENC_{m}^{RLL}(\bm{x})\in \Sigma_q^{n}$ is an $m$-RLL string}
        
        \State \textbf{Initialization:} Let $\bm{y}= \bm{x}\circ h_1(x_{n-1})$, $i=1$, and $i_{end}\triangleq n-1-m$

        \While{$i\leq i_{end}$}
        \State Let $j>i$ be the smallest index such that $y_j\not\in \{x_i,\overline{x}_i\}$
            \While{$j\geq i+m+1$}
                \State $\bm{y}\leftarrow \bm{y}_{[1,i]} \circ \bm{y}_{[i+m+1,n]} \circ Rep_{q,m_2}\big(\varphi(\bm{y}_{[i,i+m]})-2^{m} \big)\circ Rep_{q,m_1}(i)$ 
                {\color{gray}\Comment{$Rep_{q,m_1}(i)$ represents the $q$-ary representation of $i$ of length $m_1$}}
                \State $\bm{y}\leftarrow \bm{y}\circ h_2(y_{n-1})$
                \State $j\leftarrow j-m$
                \State $i_{end}\leftarrow i_{end}-m$
            \EndWhile
            \State $i\leftarrow j$
        \EndWhile
        \State \Return $\bm{y}$
    \end{algorithmic}
\end{algorithm}

\begin{algorithm}
    \caption{$m$-RLL Decoder for $m= m_1+ m_2+1$, where $m_1=\lceil \log_q n\rceil$ and $m_2=\big\lceil \lceil \log_q n+1\rceil\cdot \log_{q/2} 2 \big\rceil$}
    \label{alg:rll_dec}
    \begin{algorithmic}[1]
        \Require {$\bm{y}=ENC_{m}^{RLL}(\bm{x})\in \Sigma_q^n$ for some $\bm{x}\in \Sigma_q^{n-1}$}
        \Ensure {$\bm{x}=DEC_{m}^{RLL}(\bm{y})$}

        \While{$2| y_n$}
            \State $i\leftarrow Rep_{q,m_1}^{-1}(\bm{y}_{[n-m_1,n-1]})$
            \State $a\leftarrow Rep_{q,m_2}^{-1}(\bm{y}_{[n-m_1-m_2,n-m_1-1]})+2^{m}$
            \State $\bm{y}\leftarrow \bm{y}_{[1,i-1]} \circ \varphi^{-1}(a,y_{i}) \circ \bm{y}_{[i+1,n-m]}$
        \EndWhile
        \State $\bm{x}\leftarrow \bm{y}_{[1,n-1]}$
        \State \Return $\bm{x}$
    \end{algorithmic}
\end{algorithm}

We are now ready to present our encoder and decoder.
\begin{theorem}
  Let $q\geq 4$ and $m= m_1+ m_2+1$ with $m_1=\lceil \log_q n\rceil$ and $m_2=\big\lceil \lceil \log_q n+1\rceil\cdot \log_{q/2} 2 \big\rceil$. Algorithms \ref{alg:rll_enc} and \ref{alg:rll_dec} are correct, i.e., for any $\boldsymbol{x} \in \Sigma_q^{n-1}$, it holds that
  \begin{itemize}
      \item $ENC_{m}^{RLL}(\boldsymbol{x})\in \Sigma_q^n$ is an $m$-RLL string;
      \item $DEC_{m}^{RLL}\big(ENC_{m}^{RLL}(\bm{x}) \big)= \boldsymbol{x}$.
  \end{itemize}
  Moreover, Algorithms \ref{alg:rll_enc} and \ref{alg:rll_dec} run in time $O(n)$. 
\end{theorem}

\begin{proof}
    We first examine the correctness of Algorithm \ref{alg:rll_enc}. In Step 1, an odd symbol $h_1(x_{n-1})\notin \{x_{n-1},\overline{x}_{n-1}\}$ is appended to the end of $\bm{x}$, yielding a string $\bm{y}$ of length $n$. In the loop from Steps 2 to 11, we identify the leftmost run whose length exceeds $m$. Let $i$ and $j-1$ denote the starting and ending positions of this run, respectively. 
    Within the inner loop (Steps 4 to 9), we shorten the selected run by $m$ while appending a string of length $m$ at the end to preserve the overall length, continuing until the run length is at most $m$.
    Concretely, we first remove the substring $\bm{y}_{[i+1,i+m]}$ (note that $y_i$ is preserved) and append at the end a length-$m_1$ string to encode the information of $\bm{y}_{[i,i+m]}$, a length-$m_2$ string to encode the position of $i$, and a marker symbol to indicate loop iteration and to prevent the creation of new runs with length exceeding $m$.
    Since $i<n\leq q^{\lceil \log_q n\rceil}=q^{m_1}$, the position of $i$ can be encoded by $Rep_{q,m_1}(i)$.
    Moreover, since $\bm{y}_{[i,i+m]}\in y_i^{\oplus}$ has length $m+1$ and $y_i$ is preserved, by Definition \ref{def:indicator}, we can use the bijection $\varphi(\cdot)$ to map it into $\varphi(\bm{y}_{[i,i+m]})\in [2^{m},2^{m+1}-1]$.
    As the value of $m$ is known a priori, it suffices to record $\varphi(\bm{y}_{[i,i+m]})-2^m\in [0,2^m-1]$.
    Given that $m_2= \big\lceil \lceil \log_q n+1\rceil\cdot \log_{q/2} 2 \big\rceil$, by Lemma \ref{lem:m}, we can encode the information of $\bm{y}_{[i,i+m]}$ by $Rep_{q,m_1}\big(\varphi(\bm{y}_{[i,i+m]})-2^m\big)$.
    In Step 5, we update $\bm{y}\leftarrow \bm{y}_{[1,i]} \circ \bm{y}_{[i+m+1,n]} \circ Rep_{q,m_2}\big(\varphi(\bm{y}_{[i,i+m]})-2^{m} \big) \circ  Rep_{q,m_1}(i)$.
    Finally, the marker is defined by the even symbol $h_2(y_{n-1})\notin \{y_{n-1},\overline{y}_{n-1}\}$, and in Step 6 we update $\bm{y}\leftarrow \bm{y}\circ h_2(y_{n-1})$.
    Since the length of the selected run and the original string decrease by $m$, in Step 7 we update $j\leftarrow j-m$ and in Step 8 we update $i_{end}\leftarrow i_{end}-m$.
    When the length of the selected run reduced to at most $m$, we exit the inner loop (Steps 4–9), locate the next run whose length exceeds $m$, and repeat the process until no such run exists. 
    When the encoder finishes, the string $\bm{y}=ENC_{m}^{RLL}(\bm{x})\in \Sigma_q^{n}$ has length $n$ and for $i\leq i_{end}$, the substring $\bm{y}_{[i,i+m]}$ is not a run.
    Moreover, for $i> i_{end}$, since the substring $\bm{y}_{[i,i+m]}$ contains a marker symbol (either $h_1(\cdot)$ or $h_2(\cdot)$) and the symbol adjacent to its left, by Lemma \ref{lem:h}, we can conclude that $\bm{y}_{[i,i+m]}$ is still not a run.
    Therefore, the final output is an $m$-RLL string.
    It follows from the parity of $y_n$ that the loop (Steps 4 to 9) in Algorithm \ref{alg:rll_enc} is executed or not. Since Algorithm \ref{alg:rll_dec} is simply the inverse of the encoder, we have $DEC_{m}^{RLL}\big(ENC_{m}^{RLL}(\bm{x}) \big)= \boldsymbol{x}$.
   
    In what follows, we analyze the time complexity of the algorithms.
    Suppose there are $r$ runs $\bm{x}_{[i_1,j_1-1]}, \bm{x}_{[i_2,j_2-1]},\ldots, \bm{x}_{[i_r,j_r-1]}$ whose lengths exceed $m$. First, we can locate these runs in $O(n)$ time.
    For each $s\in [1,r]$, the inner loop (Steps 4–9) of Algorithm \ref{alg:rll_enc} is applied to transform the run $\bm{x}_{[i_s,j_s-1]}$ into a new run of length at most $m$. 
    Since each execution shortens the length of the selected run by $m$, the loop for the $s$-th runs executes $\lfloor \frac{j_s-i_s}{m}\rfloor$ times. Since the runs are disjoint, this inner loop executes
     \[
        \sum_{s=1}^r \left\lfloor \frac{j_s-i_s}{m} \right\rfloor\leq \frac{1}{m} \sum_{s=1}^r (j_s-i_s)\leq \frac{n}{m}
     \] 
     times in total.
     Moreover, in each iteration we compute the functions $Rep_{q,m_2}\big(\varphi(\cdot)-2^{m} \big)$ and $Rep_{q,m_1}(\cdot)$, which can be performed in $O(m)$ time. Therefore, the total complexity of Algorithm \ref{alg:rll_enc}  is $O(n)+O(\frac{n}{m})\cdot O(m) =O(n)$. Since Algorithm \ref{alg:rll_dec} is simply the inverse of the encoder, its complexity is also $O(n)$.
     This completes the proof.
\end{proof}

We now present an example to illustrate our encoding and decoding algorithms.

\begin{example}
  Let $n=20$ and $q=4$ with $\overline{0}=1$ and $\overline{2}=3$. Then $m_1=\lceil \log_q n\rceil= 3$, $m_2=\big\lceil \lceil \log_q n+1\rceil\cdot \log_{q/2} 2 \big\rceil=4$, and $m=m_1+m_2+1=8$. Suppose $\bm{x}=0122222222233333333\in \Sigma_q^{n-1}$. We apply Algorithm \ref{alg:rll_enc} to encode it into an $m$-RLL string.
  \begin{itemize}
    \item Initially, we compute $h_1(x_{n-1})=2\cdot \lfloor \frac{x_{n-1}}{2} \rfloor +3 \pmod{q}=2\cdot \lfloor \frac{3}{2} \rfloor +3 \pmod{4}=1$. Then we let $\bm{y}=\bm{x}\circ h_1(x_{n-1})=01222222222333333331$, $i=1$, and $i_{end}=n-1-m=20-1-8=11$. Now, we enter the loop (Steps 2-10).
    \item Since $i<i_{end}$, we set $j=3$, which represents the smallest index such that $y_{j}\notin \{y_i,\overline{y_i}\}=\{0,1\}$. Now, $i+m+1=1+8+1=10>j$. In Step 10, we update $i\leftarrow j$, i.e., we reset $i=3$.
    \item Since $i<i_{end}$, we reset $j=20$. Now, $i+m+1=3+8+1=12\leq j$ and we enter the inner loop (Steps 4-9).
    \begin{itemize}
      \item Consider the substring $\bm{y}_{[i,i+m]}=222222222$, we compute $\mathbbm{1}(\bm{y}_{[i,i+m]})= 111111111$, $\varphi(\bm{y}_{[i,i+m]})-2^m=127-64=63$, $Rep_{q,m_2}\big(\varphi(\bm{y}_{[i,i+m]})-2^m \big)= Rep_{4,4}(63)=0333$, and $Rep_{q,m_1}(i)= Rep_{4,3}(3)=003$. In Step 5, we update $\bm{y}\leftarrow \bm{y}_{[1,i]} \circ \bm{y}_{[i+m+1,n]} \circ Rep_{q,m_2}\big(\varphi(\bm{y}_{[i,i+m]})-2^m \big) \circ Rep_{q,m_1}(i)$, i.e., we reset $\bm{y}=0123333333310333003$. We then compute $h_2(y_{n-1})=2\cdot \lfloor \frac{y_{n-1}}{2} \rfloor +2 \pmod{q}=2\cdot \lfloor \frac{3}{2} \rfloor +2 \pmod{4}=0$. In Step 6, we update $\bm{y}\leftarrow \bm{y}\circ h_2(y_{n-1})$, i.e., we reset $\bm{y}=01233333333103330030$. In Step 7, we update $j\leftarrow j-m$, i.e., we reset $j=20-8=12$. In Step 8, we update $i_{end}\leftarrow i_{end}-m$, i.e., we reset $i_{end}=11-8=3$. 
      \item Since $i+m+1=3+8+1=12\leq j$, we consider the substring $\bm{y}_{[i,i+m]}=233333333$. We compute $\mathbbm{1}(\bm{y}_{[i,i+m]})= 100000000$, $\varphi(\bm{y}_{[i,i+m]})-2^m=64-64=0$, $Rep_{q,m_2}\big(\varphi(\bm{y}_{[i,i+m]})-2^m \big)= Rep_{4,4}(0)=0000$, and $Rep_{q,m_1}(i)= Rep_{4,3}(3)=003$. In Step 5, we update $\bm{y}\leftarrow \bm{y}_{[1,i]} \circ \bm{y}_{[i+m+1,n]} \circ Rep_{q,m_2}\big(\varphi(\bm{y}_{[i,i+m]})-2^m \big) \circ Rep_{q,m_1}(i)$, i.e., we reset $\bm{y}=0121033300300000003$. We then compute $h_2(y_{n-1})=2\cdot \lfloor \frac{y_{n-1}}{2} \rfloor +2 \pmod{q}=2\cdot \lfloor \frac{3}{2} \rfloor +2 \pmod{4}=0$. In Step 6, we update $\bm{y}\leftarrow \bm{y}\circ h_2(y_{n-1})$, i.e., we reset $\bm{y}=01210333003000000030$. In Step 7, we update $j\leftarrow j-m$, i.e., we reset $j=12-8=4$. In Step 8, we update $i_{end}\leftarrow i_{end}-m$, i.e., we reset $i_{end}=3-8=-5$. 
      \item Since $i+m+1=3+8+1=12> j$, we exit the inner loop (Steps 4-9).
    \end{itemize}
    \item Since $i>i_{end}$, we exit the loop (Steps 2-10) and the algorithm outputs $\bm{y}=01210333003000000030$.
  \end{itemize} 
  
  Now, we use Algorithm \ref{alg:rll_dec} to decode $\bm{x}$ from $\bm{y}=01210333003000000030$.
  \begin{itemize}
    \item Since $y_n=0$ is even, we compute $i=Rep_{q,m_1}^{-1}(\bm{y}_{[n-m_1,n-1]})= Rep_{4,3}^{-1}(\bm{y}_{[17,19]})=Rep_{4,3}^{-1}(003)=3$, $a= Rep_{4,4}^{-1}(\bm{y}_{[13,16]})+2^{4} =Rep_{4,4}^{-1}(0000)+2^{8}=0+64=64$, and $\varphi^{-1}(a,y_{i})= \varphi^{-1}(64,2)=233333333$. In Step 4, we update $\bm{y}\leftarrow \bm{y}_{[1,i-1]} \circ \varphi^{-1}(a,y_{i}) \circ \bm{y}_{[i+1,n-m]}$, i.e., we reset $\bm{y}=01233333333103330030$.
    \item Since $y_n=0$ is even, we compute $i=Rep_{q,m_1}^{-1}(\bm{y}_{[n-m_1,n-1]})= Rep_{4,3}^{-1}(\bm{y}_{[17,19]})=Rep_{4,3}^{-1}(003)=3$, $a= Rep_{4,4}^{-1}(\bm{y}_{[13,16]})+2^{4} =Rep_{4,4}^{-1}(0333)+2^{8}=63+64=127$, and $\varphi^{-1}(a,y_{i})= \varphi^{-1}(127,2)=222222222$. In Step 4, we update $\bm{y}\leftarrow \bm{y}_{[1,i-1]} \circ \varphi^{-1}(a,y_{i}) \circ \bm{y}_{[i+1,n-m]}$, i.e., we reset $\bm{y}=01222222222333333331$.
    \item Since $y_n=1$ is odd, we exit the loop and the algorithm outputs $\bm{x}=\bm{y}_{[1,n-1]}=0122222222233333333$.
  \end{itemize}
\end{example}

\subsection{Code Constructions}

In this subsection, we present two constructions of codes capable of correcting $t$ length-one reverse-complement duplications. The first construction achieves a redundancy of $2t\log_q n + O(\log_q\log_q n)$ with encoding complexity $O(n)$ and decoding complexity $O\big(n(\log_2 n)^4\big)$. The second construction achieves an improved redundancy of $(2t-1)\log_q n + O(\log_q\log_q n)$, but with encoding and decoding complexities of $O\big(n \cdot \mathrm{poly}(\log_2 n)\big)$. Both constructions combine an indel-correcting code from \cite{Li-23-ISIT} with a substitution-correcting code from \cite{Liu-24-ISIT}.

\begin{lemma}\cite[Corollary 2]{Li-23-ISIT}\label{lem:insertion}
    Let $t$ be a fixed positive integer. There exists an explicit hash function $\bm{\eta}_{n,q,t}:\Sigma_q^n \rightarrow \Sigma_q^{4t\log_q n + o(\log_q n)}$, computable in polynomial time, i.e., $\mathrm{poly}(n)$ time, such that given $\bm{\eta}_{n,q,t}(\bm{c})$ and an arbitrary string $\bm{c}'$ obtained from $\bm{c}\in \Sigma_q^n$ via $t$ insertions and deletions, one can decode $\bm{c}$ in polynomial time.
    When the length $n$ and the alphabet size $q$ are clear from the context, we abbreviate $\boldsymbol{\eta}_{n,q,t}$ as $\boldsymbol{\eta}_{t}$ to emphasize only the error-correction capability.
\end{lemma}

\begin{lemma}\cite[Sections III and IV]{Liu-24-ISIT}\label{lem:substitution}
    Let $t$ be a fixed integer, and $\ell$ be the smallest prime satisfying $\ell\geq \max\{n,2t(q-1)+1\}$.\footnote{By the well-known Bertrand-Chebyshev theorem, we have $\ell\leq 2\cdot \max\{n,2t(q-1)+1\}$.} Then there exists an explicit hash function $\bm{\zeta}_{n,q,t}: \Sigma_{q}^{n}\rightarrow \Sigma_q^{\lceil \log_q (2t(q-1)+1)\ell^{2t-1} \rceil}$, computable in $O(n)$ time, such that given $\bm{\zeta}_{n,q,t}(\bm{c})$ and an arbitrary string $\bm{c}'$ obtained from $\bm{c}\in \Sigma_q^n$ via at most $t$ substitutions, one can decode $\bm{c}$ in $O\big(n(\log_2 n)^4\big)$ time.
    When the length $n$ and the alphabet size $q$ are clear from the context, we abbreviate $\boldsymbol{\zeta}_{n,q,t}$ as $\boldsymbol{\zeta}_{t}$ to emphasize only the error-correction capability.
\end{lemma}

\subsubsection{The First Construction}
In the first construction, we simply use the function $\varphi(\cdot)$ (defined in Definition \ref{def:indicator}) to record each run and treat duplications occurring within a run as a single substitution. We begin by introducing the necessary definitions.

\begin{definition}
  Let  $\varphi(\cdot)$ be defined in Definition \ref{def:indicator}.
  For any string $\bm{x}\in \Sigma_q^n$, we define its \emph{associated vector} as 
  \[
    \bm{\phi}(\bm{x})= \big(\phi(\bm{x})_1, \phi(\bm{x})_2,\ldots, \phi(\bm{x})_{r(\bm{x})}\big),
  \]
  where $\phi(\bm{x})_i= \varphi\big(\bm{\sigma}(\bm{x},i)\big)$ and $\bm{\sigma}(\bm{x},i)$ denotes the $i$-th run of $\bm{x}$ for $i\in [1,r(\bm{x})]$.
\end{definition}

\begin{lemma}\label{lem:inverse}
  For any string $\bm{x}\in \Sigma_q^n$, given its signature $\bm{\sigma}(\bm{x})$, the inverse mapping of $\bm{\phi}$ defined by
  \[
    \bm{\phi}^{-1}\big(\bm{\phi}(\bm{x}),\bm{\sigma}(\bm{x})\big)= \big( \varphi^{-1}(\phi(\bm{x})_1, \sigma(\bm{x})_1), \varphi^{-1}(\phi(\bm{x})_2, \sigma(\bm{x})_2),\ldots, \varphi^{-1}(\phi(\bm{x})_{r(\bm{x})}, \sigma(\bm{x})_{r(\bm{x})}) \big)
  \]
  is well-defined, i.e., $\bm{\phi}^{-1}\big(\bm{\phi}(\bm{x}),\bm{\sigma}(\bm{x})\big)=\bm{x}$. Moreover, the function $\bm{\phi}^{-1}(\cdot,\cdot)$ can be computed in $O(n)$ time.
\end{lemma}

\begin{proof}
  Let $n_i= |\bm{\sigma}(\bm{x},i)|$ for $i\in [1,r(\bm{x})]$. Then $\sum_{i=1}^{r(\bm{x})}n_i=n$. By Equation (\ref{eq:inverse}), we can compute $\varphi^{-1}(\phi(\bm{x})_1, \sigma(\bm{x})_1)=\bm{\sigma}(\bm{x},i)$ in $O(n_i)$ time. Therefore, we have $\bm{\phi}^{-1}\big(\bm{\phi}(\bm{x}),\bm{\sigma}(\bm{x})\big)=\bm{x}$ and the function $\bm{\phi}^{-1}(\cdot,\cdot)$ can be computed in $\sum_{i=1}^{r(\bm{x})} O(n_i)=O(n)$ time. 
  This completes the proof.
\end{proof}

We now provide our first construction.
Let $m=\lceil \log_q n\rceil+\big\lceil \lceil \log_q n+1\rceil\cdot \log_{q/2} 2 \big\rceil+1$.
For any $\bm{x}\in \Sigma_q^{n-1}$, we first use Algorithm \ref{alg:rll_enc} to encode it into an $m$-RLL string $\bm{x}'\triangleq ENC_m^{RLL}(\bm{x})\in \Sigma_q^n$. 
Then its associated vector $\bm{\phi}(\bm{x}')= \big(\phi(\bm{x}')_1, \phi(\bm{x}')_2,\ldots, \phi(\bm{x}')_{r(\bm{x}')}\big)$, where $\phi(\bm{x}')_i= \varphi\big(\bm{\sigma}(\bm{x}',i)\big)\in [2^{m-1},2^m-1]$ for $i\in [1,r(\bm{x}')]$ can be viewed as a length-$r(\bm{x}')$ string over an alphabet of size $2^m$.
Since $r(\bm{x}')\leq n$ depends on the choice for specific $\bm{x}'$, we define 
\begin{equation*}
    \bm{\phi}'(\bm{x}') \triangleq \left( \underbrace{0, \ldots, 0}_{n-r(\bm{x}')}, \bm{\phi}(\bm{x}')\right)\in \Sigma_{2^m}^n
\end{equation*}
as the string obtained from $\bm{\phi}(\bm{x}')$ by prepending $n-r(\bm{x}')$ copies of $0$.
Since $q\geq 4$, we have $2^m\leq 2^{2\log_4 n+4}=16n$.

\begin{theorem}\label{thm:constr1}
    Let $t\geq 1$ be a fixed integer. Set $m=\lceil \log_q n\rceil+\big\lceil \lceil \log_q n+1\rceil\cdot \log_{q/2} 2 \big\rceil+1$.
    Let $\bm{R}_{t+1}(\cdot)$ be the encoding function of the $(t+1)$-fold repetition code. For any $\bm{x}\in \Sigma_q^{n-1}$, let $\bm{x}'=ENC_m^{RLL}(\bm{x})$, we define
    \[
    \mathcal{E}(\bm{x})= \left(\bm{x}', \underbrace{\overline{x_n'}, \ldots, \overline{x_n'}}_{t}, \bm{\zeta}_t\big(\bm{\phi}'(\bm{x}')\big), \bm{R}_{t+1}\big(\bm{\eta}_t(\bm{\zeta}_t(\bm{\phi}'(\bm{x}')))\big) \right).
    \]
    Then the code
    \[
    \mathcal{C}=\left\{\mathcal{E}(\bm{x}):\bm{x}\in \Sigma_q^{n-1}\right\}
    \]
    can correct $t$ length-one reverse-complement duplications with $2t\log_q n+O(\log_q\log_q n)$ redundant symbols.
    Moreover, the encoding complexity is $O(n)$ and the decoding complexity is $O\big(n(\log_2n)^4\big)$.
\end{theorem}

\begin{proof}
    For any sequence $\bm{x}\in \Sigma_2^{n-1}$, we have $n_1\triangleq \big|\bm{\zeta}_t\big(\bm{\phi}'(\bm{x}')\big) \big|= 2t\log_q n+O(1)$ and $n_2\triangleq \big|\bm{R}_{t+1}\big(\bm{\eta}_t(\bm{\zeta}_t(\bm{\phi}'(\bm{x}')))\big)\big|= O(\log_q\log_q n)$.
    Thus, the redundancy of $\mathcal{C}$ is $n_1+n_2+t+1=2t\log_q n+O(\log_q\log_q n)$.
    Note that the functions $ENC_m^{RLL}(\cdot)$, $\bm{\phi}'(\cdot)$, $\bm{\zeta}_t(\cdot)$, $\bm{\eta}_t(\cdot)$, and $\bm{R}_{t+1}(\cdot)$ can be computed in $O(n)$ time, so the overall encoding complexity is $O(n)$.
    
    Suppose $\mathcal{E}(\boldsymbol{x})$ suffers $t$ length-one reverse-complement duplications and yields in $\boldsymbol{y}$. Then 
    \begin{itemize}
        \item $\boldsymbol{y}_{[1,n+t]}$ is obtained from $\bm{x}'$ via $t$ reverse-complement duplications, as $\mathcal{E}(x)_i=\overline{x_n'}$ for $i\in [n+1,n+t]$. By Lemma \ref{lem:transform} and the definition of $\bm{\phi}'(\cdot)$, we have $r(\boldsymbol{y}_{[1,n+t]})=r(\bm{x}')$, $\bm{\sigma}(\boldsymbol{y}_{[1,n+t]})=\bm{\sigma}(\boldsymbol{x}')$, and  $\bm{\phi}'(\bm{y}_{[1,n+t]})$ is obtained from $\bm{\phi}'(\bm{x}')$ via at most $t$ substitutions. Here, if some run of $\boldsymbol{y}_{[1,n+t]}$ has length exceeding $m$, we can conclude that duplications occurred at this run and set the corresponding entry in $\boldsymbol{\phi}'(\boldsymbol{y}_{[1,n+t]})$ to zero.
        \item $\boldsymbol{y}_{[n+t+1,n+2t+n']}$ is obtained from $\bm{\zeta}_t\big(\bm{\phi}'(\bm{x}')\big)$ via $t$ insertions;
        \item $\boldsymbol{y}_{[n+t+n_1+1,n+2t+n_1+n_2]}$ is obtained from $ \bm{R}_{t+1}\big(\bm{\eta}_t(\bm{\zeta}_t(\bm{\phi}'(\bm{x}')))\big)$ via $t$ insertions.
    \end{itemize}
    Firstly, since $\bm{R}_{t+1}(\cdot)$ is the encoding function of the $(t+1)$-fold repetition code, we can recover $\bm{\eta}_t\big(\bm{\zeta}_t(\bm{\phi}'(\bm{x}'))\big)$ from $\boldsymbol{y}_{[n+t+n_1+1,n+2t+n_1+n_2]}$ in $O(n_2)=o(n)$ time.
    Next, by Lemma \ref{lem:insertion}, we can recover $\bm{\zeta}_t\big(\bm{\phi}'(\bm{x}')\big)$ from $\boldsymbol{y}_{[n+t+1,n+2t+n']}$ and $\bm{\eta}_t\big(\bm{\zeta}_t(\bm{\phi}'(\bm{x}'))\big)$ in $\mathrm{poly}(n_1)=o(n)$ time.
    Note that we can compute $r(\boldsymbol{y}_{[1,n+t]})$, $\bm{\sigma}(\boldsymbol{y}_{[1,n+t]})$, $\bm{\phi}'(\boldsymbol{y}_{[1,n+t]})$ in $O(n)$ time. 
    By Lemma \ref{lem:substitution}, we can recover $\bm{\phi}'(\boldsymbol{x}')$ from $\bm{\phi}'(\boldsymbol{y}_{[1,n+t]})$ and $\bm{\eta}_t\big(\bm{\zeta}_t(\bm{\phi}'(\bm{x}'))\big)$ in $O\big(n(\log_2 n)^4\big)$ time.
    Then we can recover $\bm{\phi}(\boldsymbol{x}')$ from $\bm{\phi}'(\boldsymbol{x}')$ and $r(\bm{y}_{[1,n+t]})$ in $O(n)$ time.
    As a result, by Lemma \ref{lem:inverse}, we can use the inverse of $\bm{\phi}(\cdot)$ to determine $\bm{x}'=\bm{\phi}^{-1}\big(\bm{\phi}(\boldsymbol{x}'),\bm{\sigma}(\boldsymbol{x}') \big)$ in $O(n)$ time. 
    Finally, we use Algorithm \ref{alg:rll_dec} to decode $\bm{x}$ from $\bm{x}'$ in $O(n)$ time.
    Thus, $\mathcal{C}$ can correct $t$ length-one reverse-complement duplications and the overall decoding complexity is $O\big(n(\log_2 n)^4\big)$.
    This completes the proof.
\end{proof}

\subsubsection{The Second Construction}

In the second construction, instead of recording the entire run, we store only a hash value of the run to aid in correcting duplications within it. This reduces the alphabet size from $O(n)$ to $O(\log_2 n)$. As a result, when using the substitution-correcting code from Lemma \ref{lem:substitution} to correct $t$ substitutions, the redundancy becomes $(2t-1)\log_q n+O(\log_q\log_q n)$.
In what follows, we detail this approach.

Let $m=\lceil \log_q n\rceil+\big\lceil \lceil \log_q n+1\rceil\cdot \log_{q/2} 2 \big\rceil+1$.
For any $\bm{x}\in \Sigma_q^{n-1}$, we first use Algorithm \ref{alg:rll_enc} to encode it into an $m$-RLL string $\bm{x}'\triangleq ENC_m^{RLL}(\bm{x})\in \Sigma_q^n$. 
Recall that $\bm{\sigma}(\bm{x}',i)$ denotes the $i$-th run of $\bm{x}'$ and $\mathbbm{1}\big(\bm{\sigma}(\bm{x}',i)\big)$ denotes the indicator vector of $\bm{\sigma}(\bm{x}',i)$, we first prepend sufficient zeros to $\mathbbm{1}\big(\bm{\sigma}(\bm{x}',i)\big)$ to increase its length to exactly $m+t$, then use the hash function $\bm{\eta}_{m+t,2,2t}(\cdot)$ (defined in Lemma \ref{lem:insertion}) to encode it, and finally record the integer representation of this hash value, denoted by $f\big(\bm{\sigma}(\bm{x}',i)\big)\in \Sigma_{2^{8t\log_2 m+o(\log_2 m)}}$. Note that $2^{8t\log_2 m+o(\log_2 m)}=m^{8t+o(1)}$.

\begin{lemma}\label{lem:f}
  Let $f(\cdot)$ be defined above. Given $f\big(\bm{\sigma}(\bm{x}',i)\big)$ and an arbitrary string $\bm{z}$ obtained from $\bm{\sigma}(\bm{x}',i)$ via at most $t$ length-one reverse-complement duplications, one can recover $\bm{\sigma}(\bm{x}',i)$ in $O\big(\mathrm{poly}(m)\big)$ time.
\end{lemma}

\begin{proof}
  Let $\bm{c}$ be obtained from $\mathbbm{1}\big(\bm{\sigma}(\bm{x}',i)\big)$ by prepending $m+t-\big|\mathbbm{1}\big(\bm{\sigma}(\bm{x}',i)\big)\big|$ zeros, and $\bm{c}'$ be obtained from $\mathbbm{1}(\bm{z})$ by prepending $m+t-|\mathbbm{1}(\bm{z})|$ zeros. Since $\bm{z}$ is obtained from $\bm{\sigma}(\bm{x}',i)$ via $t'\leq t$ length-one reverse-complement duplications, we can conclude that $\bm{c}'$ is obtained from $\bm{c}$ via $t'$ insertions and then $t'$ deletions of zeros on the leftmost side. Therefore, we can obtain $\bm{c}'$ from $\bm{c}$ via $2t$ insertions and deletions.
  Given $f\big(\bm{\sigma}(\bm{x}',i)\big)$, i.e., the integer representation of $\bm{\eta}_{m+t,2,2t}(\bm{c})$, we can compute $\bm{\eta}_{m+t,2,2t}(\bm{c})$ in $O(m)$ time. Then by Lemma \ref{lem:insertion}, we can decode $\bm{c}$ from $\bm{c}'$ and $\bm{\eta}_{m+t,2,2t}(\bm{c})$ in $O\big(\mathrm{poly}(m)\big)$ time. 
  By Definition \ref{def:indicator}, the first entry of $\mathbbm{1}\big(\bm{\sigma}(\bm{x}',i)\big)$ is $1$. Therefore, removing all the zeros on the leftmost side from $\bm{c}$ yields the correct $\mathbbm{1}\big(\bm{\sigma}(\bm{x}',i)\big)$.
  Since the duplications do not alter the first entry of $\bm{\sigma}(\bm{x}',i)$, we have $\sigma(\bm{x}',i)_1=z_1$.
  Finally, by Definition \ref{def:indicator}, we can recover $\bm{\sigma}(\bm{x}',i)$ from $\mathbbm{1}\big(\bm{\sigma}(\bm{x}',i)\big)$ and $z_1$ in $O(m)$ time.
   This completes the proof.
\end{proof}
Now, we define the vector
\begin{equation*}
    \bm{F}(\bm{x}') \triangleq \left( \underbrace{0, \ldots, 0}_{n-r(\bm{x}')}, f\big(\bm{\sigma}(\bm{x}',1))\big), \ldots, f\big(\bm{\sigma}(\bm{x}',r(\bm{x}'))\big)\right)\in \Sigma_{m^{8t+o(1)}}^n.
\end{equation*}

\begin{theorem}\label{thm:constr2}
    Let $t$ be a fixed positive integer. Set $m=\lceil \log_q n\rceil+\big\lceil \lceil \log_q n+1\rceil\cdot \log_{q/2} 2 \big\rceil+1$.
    Let $\bm{R}_{t+1}(\cdot)$ be the encoding function of the $(t+1)$-fold repetition code. For any $\bm{x}\in \Sigma_q^n$, let $\bm{x}'=ENC_m^{RLL}(\bm{x})$, we define
    \[
    \mathcal{E}(\bm{x})= \left(\bm{x}', \underbrace{\overline{x_n'}, \ldots, \overline{x_n'}}_{t}, \bm{\zeta}_t\big(\bm{F}(\bm{x}')\big), \bm{R}_{t+1}\big(\bm{\eta}_t(\bm{\zeta}_t(\bm{F}(\bm{x}')))\big) \right).
    \]
    Then the code
    \[
    \mathcal{C}=\left\{\mathcal{E}(\bm{x}):\bm{x}\in \Sigma_q^{n-1}\right\}
    \]
    can correct $t$ length-one reverse-complement duplications with $(2t-1)\log_q n+O(\log_q\log_q n)$ redundant symbols.
    Moreover, the encoding and decoding complexities are $O\big(n\cdot \mathrm{poly}(\log_2 n)\big)$.
\end{theorem}

\begin{proof}
    For any sequence $\bm{x}\in \Sigma_2^{n-1}$, we have $n_1\triangleq \big|\bm{\zeta}_t\big(\bm{F}(\bm{x}')\big) \big|= (2t-1)\log_q n+O(\log_q\log_q n)$ and $n_2\triangleq \big|\bm{R}_{t+1}\big(\bm{\eta}_t(\bm{\zeta}_t(\bm{F}(\bm{x}')))\big)\big|= O(\log_q\log_q n)$.
    Thus, the redundancy of $\mathcal{C}$ is $n_1+n_2+t+1=(2t-1)\log_q n+O(\log_q\log_q n)$.
    Note that the function $f(\cdot)$ can be computed in $O\big(\mathrm{poly}(m)\big)=O\big(\mathrm{poly}(\log_2 n)\big)$ time, we can compute the function $\bm{F}(\cdot)$ in $O\big(n\cdot \mathrm{poly}(\log_2 n)\big)$ time.
    Moreover, since the functions $ENC_m^{RLL}(\cdot)$, $\bm{\zeta}_t(\cdot)$, $\bm{\eta}_t(\cdot)$, and $\bm{R}_{t+1}(\cdot)$ can be computed in $O(n)$ time, the overall encoding complexity is $O\big(n\cdot \mathrm{poly}(\log_2 n)\big)$.
    
    Suppose $\mathcal{E}(\boldsymbol{x})$ suffers $t$ length-one reverse-complement duplications and yields in $\boldsymbol{y}$. Then 
    \begin{itemize}
        \item $\boldsymbol{y}_{[1,n+t]}$ is obtained from $\bm{x}'$ via $t$ reverse-complement duplications, as $\mathcal{E}(x)_i=\overline{x_n'}$ for $i\in [n+1,n+t]$. By Lemma \ref{lem:transform} and the definition of $\bm{F}(\cdot)$, we have $r(\boldsymbol{y}_{[1,n+t]})=r(\bm{x}')$, $\bm{\sigma}(\boldsymbol{y}_{[1,n+t]})=\bm{\sigma}(\boldsymbol{x}')$, and  $\bm{F}(\bm{y}_{[1,n+t]})$ is obtained from $\bm{F}(\bm{x}')$ via at most $t$ substitutions. 
        \item $\boldsymbol{y}_{[n+t+1,n+2t+n']}$ is obtained from $\bm{\zeta}_t\big(\bm{F}(\bm{x}')\big)$ via $t$ insertions;
        \item $\boldsymbol{y}_{[n+t+n_1+1,n+2t+n_1+n_2]}$ is obtained from $ \bm{R}_{t+1}\big(\bm{\eta}_t(\bm{\zeta}_t(\bm{F}(\bm{x}')))\big)$ via $t$ insertions.
    \end{itemize}
    Firstly, since $\bm{R}_{t+1}(\cdot)$ is the encoding function of the $(t+1)$-fold repetition code, we can recover $\bm{\eta}_t\big(\bm{\zeta}_t(\bm{F}(\bm{x}'))\big)$ from $\boldsymbol{y}_{[n+t+n_1+1,n+2t+n_1+n_2]}$ in $O(n_2)=o(n)$ time.
    Next, by Lemma \ref{lem:insertion}, we can recover $\bm{\zeta}_t\big(\bm{F}(\bm{x}')\big)$ from $\boldsymbol{y}_{[n+t+1,n+2t+n']}$ and $\bm{\eta}_t\big(\bm{\zeta}_t(\bm{F}(\bm{x}'))\big)$ in $\mathrm{poly}(n_1)=o(n)$ time.
    Note that we can compute $r(\boldsymbol{y}_{[1,n+t]})$, $\bm{\sigma}(\boldsymbol{y}_{[1,n+t]})$, $\bm{F}(\boldsymbol{y}_{[1,n+t]})$ in $O\big(n\cdot \mathrm{poly}(\log_2 n)\big)$ time.
    By Lemma \ref{lem:substitution}, we can recover $\bm{F}(\boldsymbol{x}')$ from $\bm{F}(\boldsymbol{y}_{[1,n+t]})$ and $\bm{\eta}_t\big(\bm{\zeta}_t(\bm{F}(\bm{x}'))\big)$ in $O\big(n(\log_2 n)^4\big)$ time.
    Then we can recover $f\big(\bm{\sigma}(\bm{x}',i))\big),\ldots,f\big(\bm{\sigma}(\bm{x}',r(\bm{x}')))\big)$ from $\bm{F}(\boldsymbol{x}')$ and $r(\bm{y}_{[1,n+t]})$ in $O(n)$ time.
    Moreover, for $i\in [1,r(\bm{x}')]$,  by Lemma \ref{lem:f}, we can decode $\bm{\sigma}(\bm{x}',i)$ from $\bm{\sigma}(\bm{y}_{[1,n+t]},i)$ and $f\big(\bm{\sigma}(\bm{x}',i))\big)$ in $O\big( \mathrm{poly}(\log_2 n)\big)$ time. This yields the correct $\bm{x}'$ as $\bm{x}'= \big(\bm{\sigma}(\bm{x}',1),\ldots,\bm{\sigma}(\bm{x}',r(\bm{x}')) \big)$.
    Finally, we use Algorithm \ref{alg:rll_dec} to decode $\bm{x}$ from $\bm{x}'$ in $O(n)$ time.
    Thus, $\mathcal{C}$ can correct $t$ length-one reverse-complement duplications and the overall decoding complexity is $O\big(n\cdot \mathrm{poly}(\log_2 n)\big)$.
    This completes the proof.
\end{proof}

\begin{remark}
  The constructions presented in Theorems \ref{thm:constr1} and \ref{thm:constr2} outperform the approach of directly using indel-correcting codes \cite{Li-23-ISIT, Sima-20-ISIT', Song-23-IT, Song-22-IT, Ye-25-IT}, which require $5\log_q n + O(\log_q\log_q n)$ and $(4t-1)\log_q n + O(\log_q\log_q n)$ redundant symbols when $t=2$ and $t \ge 3$, respectively.
\end{remark}

\section{Conclusion}\label{sec:concl}

In this paper, we study codes for correcting reverse-complement duplications or palindromic duplications. We construct an explicit code with a single redundant symbol capable of correcting an arbitrary number of reverse-complement duplications (respectively, palindromic duplications), provided that all duplications have length $k \ge 3\lceil \log_q n \rceil$ and are disjoint. Moreover, for $q \ge 4$, we present two explicit constructions of codes that can correct $t$ length-one reverse-complement duplications, with $2t\log_q n + O(\log_q\log_q n)$ and $(2t-1)\log_q n + O(\log_q\log_q n)$ redundant symbols, respectively.
These codes outperform general burst-insertions correcting codes in correcting reverse-complement duplications or palindromic duplications. It remains open in other parameter regimes whether codes capable of correcting palindromic or reverse-complement duplications can achieve substantially lower redundancy than directly using general burst-insertions correcting codes.

\bibliographystyle{IEEEtranS}
\bibliography{ref}

\end{document}